\documentclass[submission]{eptcs}
\providecommand{\event}{QPL 2014}
\usepackage{breakurl}

\sffamily

\usepackage{color}
\usepackage{url}
\usepackage{hyperref}
\usepackage{xypic}
\usepackage{enumitem}
\usepackage{tikz}
\usetikzlibrary{shapes,snakes}

\usepackage{nicefrac}

\def\emph{\textbf}
\def\stress{\textit}

\usepackage{amsmath,amscd,amsthm,amssymb}
\def\sfrac#1#2{\nicefrac{#1}{#2}}

\usepackage{textcomp}
\def\Vorobev{Vorob\textquotesingle\-ev}

\def\Mdot{\text{ .}}
\def\Mcomma{\text{ ,}}

\theoremstyle{plain}
\newtheorem{theorem}{Theorem}[section]
\newtheorem{proposition}[theorem]{Proposition}

\theoremstyle{definition}
\newtheorem{definition}[theorem]{Definition}

\newenvironment{calculation}{\begin{eqnarray*}&&}{\end{eqnarray*}}
\def\just#1#2{\\ &#1& \rule{2em}{0pt} \{ \mbox{\rule[-.7em]{0pt}{1.8em} \footnotesize #2 \/} \} \nonumber\\ && }
\def\ejust#1{\\ &#1& \nonumber\\ && }
\def\bimplies{\Leftrightarrow}

\newcommand{\mleft}{\mathopen{}\mathclose\bgroup\left}
\newcommand{\mright}{\aftergroup\egroup\right}

\def\setdef#1#2{\mleft\{#1\mid #2\mright\}}             
\def\enset#1{\mleft\{#1\mright\}} 

\def\Forall#1{\forall_{#1}\boldsymbol{.}\;}
\def\Exists#1{\exists_{#1}\boldsymbol{.}\;}
\def\tuple#1{\mathopen{\langle} #1 \mathclose{\rangle}}

\def\UU{\mathcal{U}}

\def\sjoin{*}
\def\dirprod{\times}
\def\Disj{\mathfrak{D}}
\def\Disjcopy#1#2{\Disj_{#1}\!\!^{\sjoin #2}}

\def\sr{\mathsf{sr}}
\def\proper#1{\pi_{#1}}

\newcommand{\ABTableNUM}[4]{
\begin{array}{cc|cccc}
  \text{A} & \text{B} & 0\,0 & 0\,1 & 1\,0 & 1\,1 \\
  \hline
  a_1 & b_1 & #1 \\ 
  a_1 & b_2 & #2 \\ 
  a_2 & b_1 & #3 \\ 
  a_2 & b_2 & #4 \\ 
\end{array}
}
\newcommand{\ACTableNUM}[4]{
\begin{array}{cc|cccc}
  \text{A} & \text{C} & 0\,0 & 0\,1 & 1\,0 & 1\,1 \\
  \hline
  a_1 & c_1 & #1 \\ 
  a_1 & c_2 & #2 \\ 
  a_2 & c_1 & #3 \\ 
  a_2 & c_2 & #4 \\ 
\end{array}
}
\newcommand{\AMTableNUM}[4]{
\begin{array}{cc|cccc}
  \text{A} & \text{M} & 0\,0 & 0\,1 & 1\,0 & 1\,1 \\
  \hline
  a_1 & m_1 & #1 \\ 
  a_1 & m_2 & #2 \\ 
  a_2 & m_1 & #3 \\ 
  a_2 & m_2 & #4 \\ 
\end{array}
}

\newcommand{\ABCTableNUM}[8]{
\begin{array}{ccc|cccccccc}
  \text{A} & \text{B} & \text{C} & 0\,0\,0 & 0\,0\,1 & 0\,1\,0 & 0\,1\,1 & 1\,0\,0 & 1\,0\,1 & 1\,1\,0 & 1\,1\,1\\
  \hline
  a_1 & b_1 & c_1 & #1 \\ 
  a_1 & b_1 & c_2 & #2 \\ 
  a_1 & b_2 & c_1 & #3 \\ 
  a_1 & b_2 & c_2 & #4 \\ 
  a_2 & b_1 & c_1 & #5 \\ 
  a_2 & b_1 & c_2 & #6 \\ 
  a_2 & b_2 & c_1 & #7 \\ 
  a_2 & b_2 & c_2 & #8 \\ 
\end{array}
}

 \def\ABCTableWstateZXmeas{
\ABCTableNUM
{  \sfrac{3}{8}  & \sfrac{1}{24} & \sfrac{1}{24} & \sfrac{1}{24}
 & \sfrac{1}{24} & \sfrac{1}{24} & \sfrac{1}{24} & \sfrac{3}{8}  }  
{  \sfrac{1}{3}  & \sfrac{1}{12} & 0             & \sfrac{1}{12}
 &  0            & \sfrac{1}{12} & \sfrac{1}{3}  & \sfrac{1}{12} } 
{  \sfrac{1}{3}  &  0            & \sfrac{1}{12} & \sfrac{1}{12} 
 &  0            & \sfrac{1}{3}  & \sfrac{1}{12} & \sfrac{1}{12} } 
{  \sfrac{1}{6}  & \sfrac{1}{6}  & \sfrac{1}{6}  &  0            
 & \sfrac{1}{6}  & \sfrac{1}{6}  & \sfrac{1}{6}  &  0            }
{  \sfrac{1}{3}  &  0            &  0            & \sfrac{1}{3}  
 & \sfrac{1}{12} & \sfrac{1}{12} & \sfrac{1}{12} & \sfrac{1}{12} } 
{  \sfrac{1}{6}  & \sfrac{1}{6}  & \sfrac{1}{6}  & \sfrac{1}{6}  
 & \sfrac{1}{6}  &  0            & \sfrac{1}{6}  &  0            } 
{  \sfrac{1}{6}  & \sfrac{1}{6}  & \sfrac{1}{6}  & \sfrac{1}{6}  
 & \sfrac{1}{6}  & \sfrac{1}{6}  &  0            &  0            } 
{   0            & \sfrac{1}{3}  & \sfrac{1}{3}  &  0            
 & \sfrac{1}{3}  &  0            &  0            &  0            } 
 }

 \def\ABCTableSvetlichny{
\ABCTableNUM
{  \sfrac{1}{4}  &       0       &       0       & \sfrac{1}{4} 
 &       0       & \sfrac{1}{4}  & \sfrac{1}{4}  &       0       }  
{  \sfrac{1}{4}  &       0       &       0       & \sfrac{1}{4} 
 &       0       & \sfrac{1}{4}  & \sfrac{1}{4}  &       0       }  
{  \sfrac{1}{4}  &       0       &       0       & \sfrac{1}{4} 
 &       0       & \sfrac{1}{4}  & \sfrac{1}{4}  &       0       }  
{        0       & \sfrac{1}{4}  & \sfrac{1}{4}  &       0      
 & \sfrac{1}{4}  &       0       &       0       & \sfrac{1}{4}  } 
{  \sfrac{1}{4}  &       0       &       0       & \sfrac{1}{4} 
 &       0       & \sfrac{1}{4}  & \sfrac{1}{4}  &       0       }  
{        0       & \sfrac{1}{4}  & \sfrac{1}{4}  &       0      
 & \sfrac{1}{4}  &       0       &       0       & \sfrac{1}{4}  } 
{        0       & \sfrac{1}{4}  & \sfrac{1}{4}  &       0      
 & \sfrac{1}{4}  &       0       &       0       & \sfrac{1}{4}  } 
{        0       & \sfrac{1}{4}  & \sfrac{1}{4}  &       0      
 & \sfrac{1}{4}  &       0       &       0       & \sfrac{1}{4}  } 
}

\def\ABCTableOtherNonSym{
\ABCTableNUM
{  \sfrac{1}{4}  & \sfrac{1}{4}  &       0       &       0       
 &       0       &       0       & \sfrac{1}{4}  & \sfrac{1}{4}  }  
{ \sfrac{1}{4}   & \sfrac{1}{4}  &       0       &       0       
 &       0       &       0       & \sfrac{1}{4}  & \sfrac{1}{4}  } 
{ \sfrac{1}{4}   & \sfrac{1}{4}  &       0       &       0      
 &       0       &       0       & \sfrac{1}{4}  & \sfrac{1}{4}  } 
{ \sfrac{1}{4}   & \sfrac{1}{4}  &       0       &       0      
 &       0       &       0       & \sfrac{1}{4}  & \sfrac{1}{4}  } 
{ \sfrac{1}{4}   & \sfrac{1}{4}  &       0       &       0      
 &       0       &       0       & \sfrac{1}{4}  & \sfrac{1}{4}  } 
{ \sfrac{1}{4}   & \sfrac{1}{4}  &       0       &       0      
 &       0       &       0       & \sfrac{1}{4}  & \sfrac{1}{4}  } 
{       0        &       0       & \sfrac{1}{4}  & \sfrac{1}{4} 
 & \sfrac{1}{4}  & \sfrac{1}{4}  &       0       &       0       } 
{       0        &       0       & \sfrac{1}{4}  & \sfrac{1}{4} 
 & \sfrac{1}{4}  & \sfrac{1}{4}  &       0       &       0       } 
}

\def\AMTableWstateZXmeasQUOT{
\AMTableNUM
{ \sfrac{5}{24} & \sfrac{1}{24}  & \sfrac{1}{24}  & \sfrac{5}{24}  }  
{ \sfrac{1}{6}  & \sfrac{1}{12}  & \sfrac{1}{6}   & \sfrac{1}{12}  } 
{ \sfrac{1}{6}  & \sfrac{1}{6}   & \sfrac{1}{12}  & \sfrac{1}{12}  } 
{ \sfrac{1}{6}  & \sfrac{1}{6}   & \sfrac{1}{6}   &      0         } 
}

\def\AMTableTotalMix{
\AMTableNUM
{ \sfrac{1}{4}  & \sfrac{1}{4}   & \sfrac{1}{4}   & \sfrac{1}{4}  }  
{ \sfrac{1}{4}  & \sfrac{1}{4}   & \sfrac{1}{4}   & \sfrac{1}{4}  } 
{ \sfrac{1}{4}  & \sfrac{1}{4}   & \sfrac{1}{4}   & \sfrac{1}{4}  } 
{ \sfrac{1}{4}  & \sfrac{1}{4}   & \sfrac{1}{4}   & \sfrac{1}{4}  } 
}

\def\ABTablePRmodel{
\ABTableNUM
{\sfrac{1}{2} &  0  &  0  & \sfrac{1}{2}}
{\sfrac{1}{2} &  0  &  0  & \sfrac{1}{2}}
{\sfrac{1}{2} &  0  &  0  & \sfrac{1}{2}}
{0   & \sfrac{1}{2} & \sfrac{1}{2} &  0 }}

\def\ACTableTotalMix{
\ACTableNUM
{ \sfrac{1}{4}  & \sfrac{1}{4}   & \sfrac{1}{4}   & \sfrac{1}{4}  }  
{ \sfrac{1}{4}  & \sfrac{1}{4}   & \sfrac{1}{4}   & \sfrac{1}{4}  } 
{ \sfrac{1}{4}  & \sfrac{1}{4}   & \sfrac{1}{4}   & \sfrac{1}{4}  } 
{ \sfrac{1}{4}  & \sfrac{1}{4}   & \sfrac{1}{4}   & \sfrac{1}{4}  } 
}

\def\AMTableOtherNonSymQUOT{
\AMTableNUM
{ \sfrac{3}{8} & \sfrac{1}{8} & \sfrac{1}{8} & \sfrac{3}{8} }  
{ \sfrac{3}{8} & \sfrac{1}{8} & \sfrac{1}{8} & \sfrac{3}{8} }  
{ \sfrac{3}{8} & \sfrac{1}{8} & \sfrac{1}{8} & \sfrac{3}{8} }  
{ \sfrac{1}{8} & \sfrac{3}{8} & \sfrac{3}{8} & \sfrac{1}{8} }  
}

\newcommand{\FigMeasInteractionSingle}{
\begin{tikzpicture}[scale=1]
  \edef\widbox{2.4}
  \edef\heibox{1.2}

  \edef\widdisplay{1.2}
  \edef\displayneedle{0.4}
  \edef\lildisplayrad{0.1}
  \edef\displaymarkerlen{0.15}

  \edef\sourcedist{4}
  \edef\sourceradius{0.5}

  \edef\dispdist{4}
  \edef\dispmaxdy{1.2}
 
  \edef\ax{0}
  \edef\ay{0}

  \path (\ax,\ay) coordinate (A); 
  \path (A) +(-\widbox/2,0) coordinate (AIn);
  \path (A) +(+\widbox/2,0) coordinate (AOut);
  \path (-\sourcedist,0) coordinate (S);
  \path (A) +(+\dispdist,\dispmaxdy) coordinate (DispZero);
  \path (A) +(+\dispdist,-\dispmaxdy-\widdisplay/2)  coordinate (DispOne);
  \path (DispZero) +(-\widdisplay/2-0.1,\widdisplay/4) coordinate (DispZeroIn);
  \path (DispOne)  +(-\widdisplay/2-0.1,\widdisplay/4) coordinate (DispOneIn);

  \shadedraw[thick, black, left color = black!20, right color = black!50]
    (\ax-\widbox/2,\ay-\heibox/2) -- (\ax+\widbox/2,\ay-\heibox/2) -- (\ax+\widbox/2,\ay+\heibox/2) -- (\ax-\widbox/2,\ay+\heibox/2) -- cycle;
  
    \path (A) node {$a$};
 
  \draw[thick] (DispZero) +(\widdisplay/2,0) arc (0:180:\widdisplay/2);
  \draw[thick] (DispZero) +(\widdisplay/2,0) -- +(-\widdisplay/2,0);
  
  \draw[black, thick] (DispZero) -- +(135:\displayneedle);
  \fill[black] (DispZero) +(180:\lildisplayrad) arc (180:0:\lildisplayrad);
  \draw[very thick, black] (DispZero) +( 45:\widdisplay/2-\displaymarkerlen) -- +( 45:\widdisplay/2);
  \draw[very thick, black] (DispZero) +(135:\widdisplay/2-\displaymarkerlen) -- +(135:\widdisplay/2);
  \path (DispZero) +( 45:\widdisplay/2+0.2) node {{\scriptsize $1$}};
  \path (DispZero) +(135:\widdisplay/2+0.2) node {{\scriptsize $0$}};
  
  \draw[thick] (DispOne) +(\widdisplay/2,0) arc (0:180:\widdisplay/2);
  \draw[thick] (DispOne) +(\widdisplay/2,0) -- +(-\widdisplay/2,0);
  
  \draw[black, thick] (DispOne) -- +(45:\displayneedle);
  \fill[black] (DispOne) +(180:\lildisplayrad) arc (180:0:\lildisplayrad);
  \draw[very thick, black] (DispOne) +( 45:\widdisplay/2-\displaymarkerlen) -- +( 45:\widdisplay/2);
  \draw[very thick, black] (DispOne) +(135:\widdisplay/2-\displaymarkerlen) -- +(135:\widdisplay/2);
  \path (DispOne) +( 45:\widdisplay/2+0.2) node {{\scriptsize $1$}};
  \path (DispOne) +(135:\widdisplay/2+0.2) node {{\scriptsize $0$}};

  \draw[orange, thick] (S) -- (AIn);
  \fill[gray] (S) circle (\sourceradius);
  \path (S) node {{\scriptsize source}};
  \draw[orange, thick] (AOut) -- (DispOneIn);
\end{tikzpicture}
}

\newcommand{\FigMeasInteractionBeam}{
\begin{tikzpicture}[scale=1]
  \edef\widbox{2.4}
  \edef\heibox{1.2}

  \edef\widdisplay{1.2}
  \edef\displayneedle{0.4}
  \edef\lildisplayrad{0.1}
  \edef\displaymarkerlen{0.15}

  \edef\sourcedist{4}
  \edef\sourceradius{0.5}

  \edef\dispdist{4}
  \edef\dispmaxdy{1.2}

  \edef\nrays{9}
  \edef\nrayszero{3}
  \edef\nraysone{6}
  \edef\distrays{.07}

  \edef\finallenght{.7}
  \edef\finalsep{.02}
  \edef\finalbarhalflen{.04}
  \edef\finaltextsep{.7}

  \edef\ax{0}
  \edef\ay{0}

  \path (\ax,\ay) coordinate (A); 
  \path (A) +(-\widbox/2,0) coordinate (AIn);
  \path (A) +(+\widbox/2,0) coordinate (AOut);
  \path (-\sourcedist,0) coordinate (S);
  \path (A) +(+\dispdist,\dispmaxdy) coordinate (DispZero);
  \path (A) +(+\dispdist,-\dispmaxdy-\widdisplay/2)  coordinate (DispOne);
  \path (DispZero) +(-\widdisplay/2-0.1,\widdisplay/4) coordinate (DispZeroIn);
  \path (DispOne)  +(-\widdisplay/2-0.1,\widdisplay/4) coordinate (DispOneIn);
  \path (DispZero) +(+\widdisplay/2+0.1,\widdisplay/4) coordinate (FinalZI);
  \path (DispOne)  +(+\widdisplay/2+0.1,\widdisplay/4) coordinate (FinalOI);
  \path (FinalZI) + (\finallenght,0) coordinate (FinalZF);
  \path (FinalOI) + (\finallenght,0) coordinate (FinalOF);

  \shadedraw[thick, black, left color = black!20, right color = black!50]
    (\ax-\widbox/2,\ay-\heibox/2) -- (\ax+\widbox/2,\ay-\heibox/2) -- (\ax+\widbox/2,\ay+\heibox/2) -- (\ax-\widbox/2,\ay+\heibox/2) -- cycle;
  
    \path (A) node {$a$};
 
  \draw[thick] (DispZero) +(\widdisplay/2,0) arc (0:180:\widdisplay/2);
  \draw[thick] (DispZero) +(\widdisplay/2,0) -- +(-\widdisplay/2,0);
  
  \draw[black, thick] (DispZero) -- +(135:\displayneedle);
  \fill[black] (DispZero) +(180:\lildisplayrad) arc (180:0:\lildisplayrad);
  \draw[very thick, black] (DispZero) +( 45:\widdisplay/2-\displaymarkerlen) -- +( 45:\widdisplay/2);
  \draw[very thick, black] (DispZero) +(135:\widdisplay/2-\displaymarkerlen) -- +(135:\widdisplay/2);
  \path (DispZero) +( 45:\widdisplay/2+0.2) node {{\scriptsize $1$}};
  \path (DispZero) +(135:\widdisplay/2+0.2) node {{\scriptsize $0$}};
  
  \draw[thick] (DispOne) +(\widdisplay/2,0) arc (0:180:\widdisplay/2);
  \draw[thick] (DispOne) +(\widdisplay/2,0) -- +(-\widdisplay/2,0);
  
  \draw[black, thick] (DispOne) -- +(45:\displayneedle);
  \fill[black] (DispOne) +(180:\lildisplayrad) arc (180:0:\lildisplayrad);
  \draw[very thick, black] (DispOne) +( 45:\widdisplay/2-\displaymarkerlen) -- +( 45:\widdisplay/2);
  \draw[very thick, black] (DispOne) +(135:\widdisplay/2-\displaymarkerlen) -- +(135:\widdisplay/2);
  \path (DispOne) +( 45:\widdisplay/2+0.2) node {{\scriptsize $1$}};
  \path (DispOne) +(135:\widdisplay/2+0.2) node {{\scriptsize $0$}};

  \pgfmathsetmacro\halfsize{\distrays*(\nrays-1)/2}
  \path (S)    +(0,\halfsize) coordinate (STop);
  \path (AIn)  +(0,\halfsize) coordinate (AInTop);
  \path (AOut) +(0,\halfsize) coordinate (AOutTop);
  \foreach \jj in {1,...,\nrays}
  {
    \pgfmathsetmacro\dd{\distrays*(\jj-1)}
    \path (STop)   +(0,-\dd) coordinate (Sjj);
    \path (AInTop) +(0,-\dd) coordinate (AInjj);
    \draw[orange,thin] (Sjj) -- (AInjj);
  }

  \fill[gray] (S) circle (\sourceradius);
  \path (S) node {{\scriptsize source}};

  \pgfmathsetmacro\halfsizeZ{\distrays*(\nrayszero-1)/2}  
  \pgfmathsetmacro\halfsizeO{\distrays*(\nraysone-1)/2}  
  \path (DispZeroIn) +(0, \halfsizeZ) coordinate (DispZeroInTop);
  \path (DispOneIn)  +(0, \halfsizeO) coordinate (DispOneInTop);
  \path (FinalZI)    +(0, \halfsizeZ) coordinate (FinalZITop);
  \path (FinalOI)    +(0, \halfsizeO) coordinate (FinalOITop);
  \path (FinalZF)    +(0, \halfsizeZ) coordinate (FinalZFTop);
  \path (FinalOF)    +(0, \halfsizeO) coordinate (FinalOFTop);
  \path (FinalZF)    +(0,-\halfsizeZ) coordinate (FinalZFBot);
  \path (FinalOF)    +(0,-\halfsizeO) coordinate (FinalOFBot);
  \foreach \jj in {1,...,\nrayszero}
  {
    \pgfmathsetmacro\dd{\distrays*(\jj-1)}
    \path (AOutTop)       +(0,-\dd) coordinate (AOutjj);
    \path (DispZeroInTop) +(0,-\dd) coordinate (Dispjj);
    \draw[orange,very thin] (AOutjj) -- (Dispjj);
    \path (FinalZITop)+(0,-\dd) coordinate (FinalIjj);
    \path (FinalZFTop)+(0,-\dd) coordinate (FinalFjj);
    \draw[orange,very thin] (FinalIjj) -- (FinalFjj);
  }
  \foreach \jj in {1,...,\nraysone}
  {
    \pgfmathsetmacro\dd{\distrays*(\jj-1)}
    \pgfmathsetmacro\ddo{\distrays*(\nrayszero+\jj-1)}
    \path (AOutTop)       +(0,-\ddo) coordinate (AOutjj);
    \path (DispOneInTop)  +(0,-\dd) coordinate (Dispjj);
    \draw[orange,very thin] (AOutjj) -- (Dispjj);
    \path (FinalOITop)+(0,-\dd) coordinate (FinalIjj);
    \path (FinalOFTop)+(0,-\dd) coordinate (FinalFjj);
    \draw[orange,very thin] (FinalIjj) -- (FinalFjj);
  }

   \path (FinalOFTop) ++(\finalsep,0) coordinate (FinalBarOltop) ++(\finalbarhalflen,0) coordinate (FinalBarOmtop) ++(\finalbarhalflen,0) coordinate (FinalBarOrtop);
   \path (FinalOFBot) ++(\finalsep,0) coordinate (FinalBarOlbot) ++(\finalbarhalflen,0) coordinate (FinalBarOmbot) ++(\finalbarhalflen,0) coordinate (FinalBarOrbot);

   \path (FinalZFTop) ++(\finalsep,0) coordinate (FinalBarZltop) ++(\finalbarhalflen,0) coordinate (FinalBarZmtop) ++(\finalbarhalflen,0) coordinate (FinalBarZrtop);
   \path (FinalZFBot) ++(\finalsep,0) coordinate (FinalBarZlbot) ++(\finalbarhalflen,0) coordinate (FinalBarZmbot) ++(\finalbarhalflen,0) coordinate (FinalBarZrbot);

   \draw (FinalBarOltop) -- (FinalBarOrtop);
   \draw (FinalBarOlbot) -- (FinalBarOrbot);
   \draw (FinalBarOmtop) -- (FinalBarOmbot);
   \draw (FinalBarZltop) -- (FinalBarZrtop);
   \draw (FinalBarZlbot) -- (FinalBarZrbot);
   \draw (FinalBarZmtop) -- (FinalBarZmbot);

   \draw (FinalOF) +(\finaltextsep,0) node {{\scriptsize $I_1 = \sfrac{2}{3}$}};
   \draw (FinalZF) +(\finaltextsep,0) node {{\scriptsize $I_0 = \sfrac{1}{3}$}};

\end{tikzpicture}
}

\title{On monogamy of non-locality and macroscopic averages: examples and preliminary results}
\def\titlerunning{On monogamy of non-locality and macroscopic averages: examples and preliminary results}
\author{
Rui Soares Barbosa
\institute{Quantum Group \\ Department of Computer Science\\
University of Oxford
}
\email{rui.soares.barbosa@cs.ox.ac.uk}
}
\def\authorrunning{Rui Soares Barbosa}

\date{\today}

\begin{document}

\maketitle

\begin{abstract}
We explore a connection between monogamy of non-locality 
and a weak macroscopic locality condition: the locality of the average behaviour.
These are revealed by our analysis as being two sides of the same coin.

Moreover, we exhibit a structural reason for both 
in the case of Bell-type multipartite scenarios, shedding light on but also generalising the results in the literature
\cite{RamanathanEtAl:LocalRealismOfMacroscopicCorrelations,PawlowskiBrukner:MonogamyOfBellIneqsInNonsigTheories}.
More specifically, we show that,
provided the number of particles in each site is large enough compared to the number of allowed measurement settings,
and whatever the microscopic state of the system,
the macroscopic average behaviour is local realistic, or equivalently, general multipartite monogamy relations hold.

This result relies on a classical mathematical theorem by \Vorobev~\cite{vorobev}
about extending compatible families of probability distributions defined on the faces of a simplicial complex --
in the language of the sheaf-theoretic framework of Abramsky \& Brandenburger \cite{AbramskyBrandenburger},
such families correspond to no-signalling empirical models, and
the existence of an extension corresponds to locality or non-contextuality.
Since \Vorobev's theorem depends solely on the structure of the simplicial complex,
which encodes the compatibility of the measurements,
and not on the specific probability distributions (i.e. the empirical models),
our result about monogamy relations and locality of macroscopic averages
holds not just for quantum theory, but for any empirical model satisfying the no-signalling condition.

  In this extended abstract, 
  we illustrate our approach by working out a couple of examples,
  which convey the intuition behind our analysis
  while keeping the discussion at an elementary level.

  ~
  \Vorobev's theorem.
\end{abstract}


\section{Introduction}\label{sec:introduction}
Bell's theorem \cite{Bell-thm} showed that the quantum world is non-local:
the correlations between the outcomes of measurements on two entangled (space-like separated)
particles are too strong to be explainable by a common `local' cause.
The usual monogamy of non-locality relations impose a limit on the amount of non-locality shared
by one party with multiple other parties.
For example, in a tripartite ($A$, $B$ and $C$) system where each experimenter has two measurement settings available,
there is a trade-off between the strengths of violation
of a Bell inequality by the subsystem composed of $A$ and $B$ and the subsystem composed of $A$ and $C$.
More explicitly, for a bipartite Bell inequality $\mathcal{B}(-,-) \leq R$,
the added inequality $\mathcal{B}(A,B) + \mathcal{B}(A,C) \leq R + R$ holds,
even if each of $\mathcal{B}(A,B) \leq R$ and $\mathcal{B}(A,C) \leq R$ might be violated.

Ramanathan et al. \cite{RamanathanEtAl:LocalRealismOfMacroscopicCorrelations} consider multipartite
macroscopic systems, consisting of a large number of particles at each site, which are described
by quantum mechanics.
At each site, only `macroscopic' measurements are available:
e.g. magnetisation along some direction, which arises as a sort of average of the individual spin measurements in that direction for each particle in the site.
The authors are concerned only with the average behaviour
over all the microscopic particles -- this can
be obtained from the mean values of intensities measured macroscopically
(see Section \ref{ssec:macro-avg} for a more detailed explanation).
They show that, whatever the quantum state of the system
(so regardless of the form and strength of the entanglement between the particles),
and provided the number of particles at each site is large enough 
when compared to the number of different measurement settings available,
there is a local realistic explanation for these macroscopic average correlations.
The reason for such classicality 
is that non-local effects are diluted by averaging due to the restrictions imposed by monogamy.

However, monogamy holds more generally than just for quantum mechanics.
Paw{\l}owsky \& Brukner \cite{PawlowskiBrukner:MonogamyOfBellIneqsInNonsigTheories} show that all no-signalling theories
satisfy monogamy relations for the violation of any bipartite Bell-type inequality.
More specifically, given a general bipartite Bell inequality $\mathcal{B}(A,B) \leq R$,
they consider a scenario with one Alice, $A$, and $k$ independent copies of Bob, $B^{(1)}, \ldots, B^{(k)}$,
with $k$ equal to the number of measurement settings available to Bob.
They show that a monogamy relation for the bipartite inequality,
$\sum_{m=1}^{k}\mathcal{B}(A,B^{(m)}) \leq k R$, is satisfied by any no-signalling theory.

The methods used in the two above-cited papers are manifestly similar.
%
%
This observation, also made in \cite{RamanathanEtAl:LocalRealismOfMacroscopicCorrelations}, leads one to conjecture that
the results about local macroscopic averages also hold in general for any no-signalling theory.
We show that this turns out to be the case:
our investigation establishes a clear structural connection between the two papers
leading to a generalisation of the results of both.
Indeed, our main result for multipartite scenarios (Proposition \ref{prop:nkr-vorobev}) can be read in two ways:
on the one hand, it generalises the result of Paw{\l}owski \& Brukner \cite{PawlowskiBrukner:MonogamyOfBellIneqsInNonsigTheories},
concerning deriving monogamy relations from the no-signalling condition,
from bipartite to multipartite Bell inequalities with an arbitrary number of sites;
on the other hand, it generalises the result of Ramanathan et al. \cite{RamanathanEtAl:LocalRealismOfMacroscopicCorrelations},
about the classicality of macroscopic average behaviour in multipartite models,
from quantum models to all no-signalling models.
Let us spell out the consequences of Proposition \ref{prop:nkr-vorobev} from each of these perspectives.
\begin{itemize}
  \item Let $\mathcal{B}(A,B,C,\ldots) \leq R$ be a general
    Bell-type inequality over $n$ sites $A, B, C, \ldots$  with respectively
    $k_A, k_B, k_C, \ldots$ measurement settings available.
    Then, consider a scenario with a single copy of one of the sites, say $A$, and with
    $r_B$ copies of site $B$, $B^{(1)}, \ldots, B^{(r_B)}$,
    $r_C$ copies of site $C$, $C^{(1)}, \ldots, C^{(r_C)}$, etc. 
    The  monogamy relation for the satisfaction of the Bell inequality,
    \begin{equation}\label{eq:monogamymulti}
      \sum_{m_B=1}^{r_B}\sum_{m_C=1}^{r_C} \cdots \;\,\mathcal{B}(A,B^{(m_B)},C^{(m_C)}, \ldots) \leq r_Br_C \cdots R \,\text{ ,}
    \end{equation}
    is satisfied by any no-signalling model if and only if the number of copies of each site is at least the number of measurement settings at that site; i.e. $r_B \geq k_B$, $r_C \geq k_C$, etc.
    Reference \cite{PawlowskiBrukner:MonogamyOfBellIneqsInNonsigTheories} addressed the particular case $n=2$
    for which it proved the `if' side of this result when $\vec{r} = \vec{k}$.
  \item From the other perspective, suppose that we have a scenario with $n$ `macroscopic' sites $A, B, C, \ldots$
    with respectively $k_A, k_B, k_C, \ldots$ measurement settings available,
    and suppose that each of these macroscopic sites is constituted by a number $r_i \; (i \in \enset{A, B, C, \ldots})$ of
    microscopic sites.
    We assume that each of the $k_i$ measurement settings available at site $i$
    corresponds to performing a similar measurement on all the microscopic sites (say, particles) that constitute it:
    the expected value of the macroscopic measurement then tells us the average behaviour among all the microscopic sites
    (see Section \ref{ssec:macro-avg} for more details).
    We show that
    if the number of microscopic sites forming each macroscopic site is at least the number of measurement settings at that site,
    i.e. if $r_i \geq k_i$ for all sites $i \in \enset{A, B, C, \ldots}$, 
    then any no-signalling empirical model on a microscopic scenario
    has local average macroscopic behaviour.
    Reference  \cite{RamanathanEtAl:LocalRealismOfMacroscopicCorrelations} proved this result, 
    but restricted to the case of quantum mechanical correlations.
    We show that having local macroscopic averages is not a particular property
    of quantum mechanics distinguishing it from super-quantum correlations.
    Note that this is not to say that we cannot distinguish them by other, more refined notions of macroscopic locality (cf. e.g. \cite{NavascuesWunderlich2009}, and see Section \ref{ssec:macro-avg} for a discussion).
\end{itemize}
Moreover,
it becomes apparent that the two items above are essentially two ways of looking at the same thing.


More important perhaps than these generalisations is that our analysis highlights
the \emph{structural reason} why these results hold.
This is related to a characterisation due to \Vorobev\ of the measurement scenarios
that are inherently local or non-contextual.
The idea is that quotienting a large scenario
by the identification
(of sites that are `copies' or instances of the same site,
or of microscopic sites forming a single macroscopic site)
\[A^{(1)} \sim \cdots \sim A^{(r_A)} \quad\quad
  B^{(1)} \sim \cdots \sim B^{(r_B)} \quad\quad
  C^{(1)} \sim \cdots \sim C^{(r_C)} \quad\quad
  \ldots  \quad \Mcomma\]
along which one considers the monogamy relation or takes the average,
yields such an inherently local scenario.
Hence, the model obtained by averaging along the symmetry,
being defined on this quotient scenario, must be local.
This also implies that the original model satisfies all monogamy relations
for this symmetry, which are simply the invariant Bell inequalities.


Another important aspect of this work is its potential for further generalisation,
as indicated in Section \ref{sec:conclusions}. For example, the same ideas can potentially be applied
to yield monogamy relations for violation of contextuality inequalities,
or to study macroscopic averages in more general scenarios as well. 

The central aim of this extended abstract is to convey the intuition behind this structural proof.
We mainly focus on showing a couple of simple examples whose geometric realisations can be easily visualised.
We try to keep the presentation at an elementary level, ignoring some of the more involved technical details,
rather to focus on the central ideas and intuitions. 
A longer version of this work,
  containing all the technical details and full proofs, as well as presenting things in greater generality,
  is under preparation.

\paragraph*{Outline.}
Section \ref{sec:setting} gives a quick overview of the main ingredients
of the sheaf-theoretic approach.
In Section \ref{sec:relating-firstexample}, after a discussion
that clarifies the meaning of average macroscopic behaviour
and compares it to other notions discussed in the literature,
we observe a connection between 
such averages and monogamy of non-locality,
in the simplest (tripartite) scenario where the latter arises, in its most familiar form.
Section \ref{sec:structuralexplanation} presents \Vorobev's theorem,
and illustrates, using the motivating example,
how it can provide a structural explanation for
monogamy relations and local macroscopic averages,
due to the acyclicity of a certain quotient complex.
We also consider another tripartite example, with more measurement settings, 
for which the quotient is not acyclic, and so the explanation above does not apply:
monogamy relations may fail to hold and macroscopic averages fail to be classical.
Both examples are particular cases of the general multipartite scenarios that are considered in Section \ref{sec:gen-multipartite},
where we present (without proof) the complete characterisation of those whose quotients are acyclic (Proposition \ref{prop:nkr-vorobev}), 
which yields the generalisations of the results of the two papers mentioned above.
Finally, Section \ref{sec:conclusions} concludes with a summary and an outlook.

\newpage

\section{Measurement scenarios and empirical models}\label{sec:setting}
We quickly summarise some of the basic ideas of
the sheaf-theoretic framework of Abramsky \& Brandenburger \cite{AbramskyBrandenburger},
which provides a unified treatment of non-locality and contextuality in the general setting of no-signalling probabilistic models.
For the purpose of the present document, we shall be mainly concerned with non-locality.
Still, the geometric structures of this framework provide an
appropriate setting in which to understand -- and visualise -- monogamy and macroscopic averages.
For some other interesting results stemming from this sheaf-theoretic approach to non-locality and contextuality,
the reader is referred to \cite{AbramskyHardy:LogicalBellIneqs,AbramskyMansfieldBarbosa:Cohomology-QPL,AbramskyConstantin:ClassificationMultipartiteStates,MansfieldBarbosa:QPL2013,Abramsky12:databases,Mansfield:DPhil-thesis}.

\subsection{Measurement scenarios}\label{ssec:meas-scenarios}
A \emph{measurement scenario}
is given by an abstract simplicial complex $\Sigma$ on the (finite) set $X$ of allowed measurements\footnote{An
abstract simplicial complex on a set (of vertices) $X$ is a family of subsets of $X$, called faces,
that is downwards-closed and contains
all the singletons $\enset{x}$, $x \in X$. 
This is interpreted as a combinatorial description of a geometrical object
given as a collage of points (the singletons), line segments (sets of two elements), triangles (sets of three elements),
and their higher-dimensional counterparts.}
(or equivalently, by a cover $\UU$ of $X$: the corresponding simplicial complex is obtained by down closure;
and conversely, the maximal faces of a simplicial complex form a cover of $X$).
Each face of the complex is called a measurement context.
The intuition is that measurements in the same context can be performed together.
Examples include multipartite Bell-type scenarios, Kochen--Specker configurations, and more.

Let us take as an example the simplest scenario in which monogamy relations arise, in their most familiar form.
We consider a Bell-type scenario with
three sites ($A$, $B$ and $C$) and two possible measurement settings available to the experimenter at each site
($a_1$ and $a_2$ for $A$, $b_1$ and $b_2$ for $B$, and $c_1$ and $c_2$ for $C$).
As usual, the choice of measurement at each site can be made independently of the other sites.
Formally, the set of available measurements is $X = \enset{a_1, a_2, b_1, b_2, c_1, c_2}$
and the cover of maximal contexts is
\[\UU =  \setdef{ \enset{a_i,b_j,c_k} } { i,j,k \in \enset{1,2} } \Mdot\]
The corresponding simplicial complex is a \stress{hollow} octahedron, depicted below:
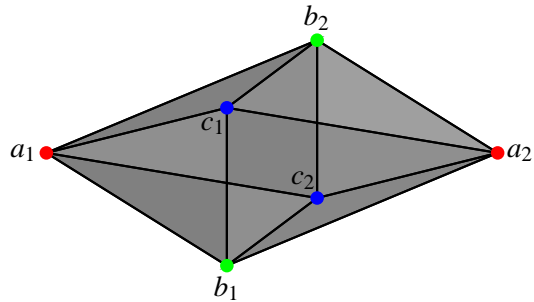
\begin{center}
\begin{tikzpicture}[scale=3]
  \path (-1 , 0 ) coordinate (A1);
  \path ( 1 , 0 ) coordinate (A2);
  \path (-.2,-.5) coordinate (B1);
  \path (+.2,+.5) coordinate (B2);
  \path (-.2,+.2) coordinate (C1);
  \path (+.2,-.2) coordinate (C2);

  \begin{scope}[opacity=0.5, color=gray]
    \fill (A1) -- (B1) -- (C1) -- cycle;
    \fill (A1) -- (B1) -- (C2) -- cycle;
    \fill (A1) -- (B2) -- (C1) -- cycle;
    \fill (A1) -- (B2) -- (C2) -- cycle;
    \fill (A2) -- (B1) -- (C1) -- cycle;
    \fill (A2) -- (B1) -- (C2) -- cycle;
    \fill (A2) -- (B2) -- (C1) -- cycle;
    \fill (A2) -- (B2) -- (C2) -- cycle;
  \end{scope}

  \begin{scope}[ thick,black]
    \draw (A1) -- (B1) -- (C1) -- cycle;
    \draw (A1) -- (B1) -- (C2) -- cycle;
    \draw (A1) -- (B2) -- (C1) -- cycle;
    \draw (A1) -- (B2) -- (C2) -- cycle;
    \draw (A2) -- (B1) -- (C1) -- cycle;
    \draw (A2) -- (B1) -- (C2) -- cycle;
    \draw (A2) -- (B2) -- (C1) -- cycle;
    \draw (A2) -- (B2) -- (C2) -- cycle;
  \end{scope}
 
  \fill[red]   (A1) circle(.03) ;
  \fill[red]   (A2) circle(.03) ;
  \fill[green] (B1) circle(.03) ;
  \fill[green] (B2) circle(.03) ;
  \fill[blue]  (C1) circle(.03) ;
  \fill[blue]  (C2) circle(.03) ;
  
  \path (A1) +(-.1 ,0  ) node {$a_1$}; 
  \path (A2) +(+.1 ,0  ) node {$a_2$};
  \path (B1) +(0   ,-.1) node {$b_1$};
  \path (B2) +(0   ,+.1) node {$b_2$};
  \path (C1) +(-.06,-.08) node {$c_1$};
  \path (C2) +(-.06,+.08) node {$c_2$};
\end{tikzpicture}
\end{center}
Note that this complex can be described in a more compositional way as
\[\Disj_2^{\sjoin 3} = \Disj_2 \sjoin \Disj_2 \sjoin \Disj_2 \Mcomma\]
where 
$\Disj_2$ is the discrete simplicial complex on two vertices\footnote{The discrete complex on $n$ vertices, $\Disj_n$,
is the minimal simplicial complex on $n$ vertices, containing no faces of dimension higher than $0$ (lines, triangles, etc.).
Formally, $\Disj_n := \enset{\emptyset, \enset{1}, \ldots, \enset{n}}$.}, corresponding to the scenario available to each experimenter,
and $\sjoin$ stands for the simplicial join operation\footnote{Given simplicial complexes
$\Sigma_1$ and $\Sigma_2$ on vertex sets $X_1$ and $X_2$ respectively,
their simplicial join is defined on the vertex set $X_1 \sqcup X_2$
as 
$\Sigma_1 \sjoin \Sigma_2 := \setdef{\sigma_1 \sqcup \sigma_2}{\sigma_1 \in \Sigma_1, \sigma_2 \in \Sigma_2} = \setdef{\sigma \subseteq X_1 \sqcup X_2}{\sigma \cap X_1 \in \Sigma_1 \;\land\; \sigma \cap X_2 \in \Sigma_2}$.
}, which captures parallel composition of scenarios.
For more details on this, see e.g. \cite{RSB:DPhil-thesis-forth}.

In the explicit examples of empirical models in the rest of this text,
we shall take all measurements to have two possible outcomes: $0$ and $1$.
This is irrelevant as none of the results we consider is sensitive to the sets of outcomes,
as long as there are at least two outcomes per measurement.
With this extra assumption, the scenario $\Disj_2^{\sjoin 3}$ above is also customarily known as the
$(3,2,2)$ scenario: the numbers stand for 3 sites, 2 measurement settings at each site, and 2 outcomes for each measurement.

\subsection{Empirical models and extendability}\label{ssec:empirical-models}
While a measurement scenario is an abstract description of a set of possible experiments,
empirical models represent particular (real or hypothetical)
probabilistic results of these experiments
(one can think of frequencies tabulated from runs of the experiments on ensembles of identically prepared systems).

Given a measurement scenario $\UU$, 
an \emph{empirical model}
is a compatible family of probability distributions $(\mu_C)_{C \in \UU}$,
where each $\mu_C$ is a distribution on joint outcomes of the measurements in context $C$.
Compatibility here means that $\mu_C$ and $\mu_{C'}$ marginalise to the same distribution on 
outcomes of measurements in $C \cap C'$.
In the case of multipartite scenarios, this corresponds to the usual \emph{no-signalling} condition.

For such an empirical model, we are concerned with the existence of a global probability distribution
$\mu_X$ on the joint outcomes of all the measurements that marginalises to all the distributions $\mu_C$.
It is shown by Abramsky \& Brandenburger \cite{AbramskyBrandenburger}
that such a global extension exists iff the model admits a non-contextual
(or local, in the particular case of multipartite scenarios) hidden variable explanation.
So, the set of joint outcomes of all measurements ($\enset{0,1}^X$ in the case of dichotomic measurements)
can be seen as a canonical hidden variable space.
Obstructions to such extensions are witnessed by violations of Bell-type inequalities by the probability distributions $\mu_C$
(cf. \cite{AbramskyHardy:LogicalBellIneqs} for a general scheme, based on logical consistency conditions, for
deriving complete sets of Bell-type inequalities on any measurement scenario).

Let us consider some examples.
Take the tripartite scenario from Section \ref{ssec:meas-scenarios}.
An empirical model for this scenario is a collection of no-signalling
probabilities of the form $p(a_i,b_j,c_k = x, y, z)$,
where $x$, $y$, $z$ range over the possible outcomes of the respective measurements.
An example of a valid empirical model is represented in the following table --
this is the model obtained
by preparing a 3-qubit system in the W state
and allowing $Z$ and $X$ measurements at each site (on each qubit).
\begin{equation}\label{extable1}
\ABCTableWstateZXmeas
\end{equation}
Another example is the super-quantum tripartite box known as Svetlichny box \cite{BarrettEtAl2005:NonlocalCorrelationsInfoResource}:
\begin{equation}\label{extable2}
\ABCTableSvetlichny
\end{equation}
Finally, let us also consider a non-symmetric example (with respect to $B$ and $C$):
\begin{equation}\label{extable3}
\ABCTableOtherNonSym
\end{equation}
All three examples above are non-local.

\section{Relating monogamy and macroscopic averages: a first example}\label{sec:relating-firstexample}
\subsection{Macroscopic average behaviour}\label{ssec:macro-avg}
We start by clarifying
what we mean by macroscopic average behaviour.
This is the same as the macroscopic correlations considered in
Ramanathan et al. \cite{RamanathanEtAl:LocalRealismOfMacroscopicCorrelations}.
We also contrast it
with the notion of macroscopic locality suggested in Navascu\'es \& Wunderlich \cite{NavascuesWunderlich2009} and Bancal et al. \cite{BancalEtAl2008}.

In order to understand the averaging process,
we first consider single-site measurements.
Let us take a (microscopic) measurement with $l$ possible outcomes.
One can imagine that, in such a measurement, a \stress{single} particle is subjected to an interaction $a$,
a measurement process of some sort, resulting in the particle colliding with one of the
$l$ detectors corresponding to the measurement outcomes (Figure \ref{fig:measinteractionsingle}).
The nature of this interaction
might be probabilistic; repeating the experiment many times on identically prepared systems
allows us to collect statistical data $p(x \mid a)$ or $p(a = x)$ (with $x \in \enset{0, \ldots, l-1}$)
corresponding to the probability of the detector $x$ being clicked given that one has decided to measure $a$.

\begin{figure}
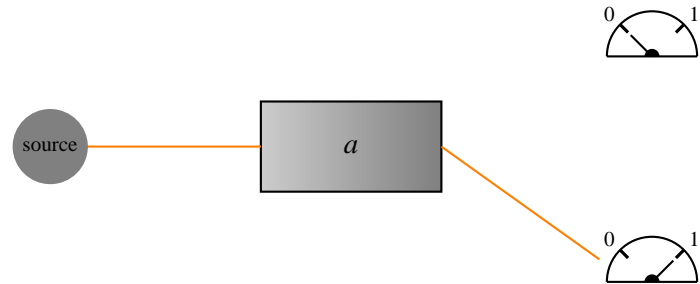

  \begin{center}
  \FigMeasInteractionSingle
  \end{center}
  \caption{Dichotomic measurement, $a$, performed on a single particle, which ends up hitting the detector corresponding to one of the possible outcomes (in this case, outcome $1$).}\label{fig:measinteractionsingle}
\end{figure}

Now we introduce a change to this setup.
In a macroscopic experiment,
the experimenter
receives a beam of $N$ particles instead of
a single particle. The same interaction is applied
(simultaneously)
to all the particles in the beam,
dividing it into smaller beams that collide with each of the detectors $0, \ldots, l-1$.
The information one can obtain from such an experiment is the number of particles that collide with each detector,
i.e. the intensity of each of the smaller resulting beams (Figure \ref{fig:measinteractionbeam}).
Note that instead of beams of photons, we could also think of regions of a magnetic material
where measuring the magnetisation in a certain direction corresponds to making a spin measurement
on all the $N$ particles in the region.
The details are not essential to the discussion, so we shall keep talking mostly about `beams'.

\begin{figure}
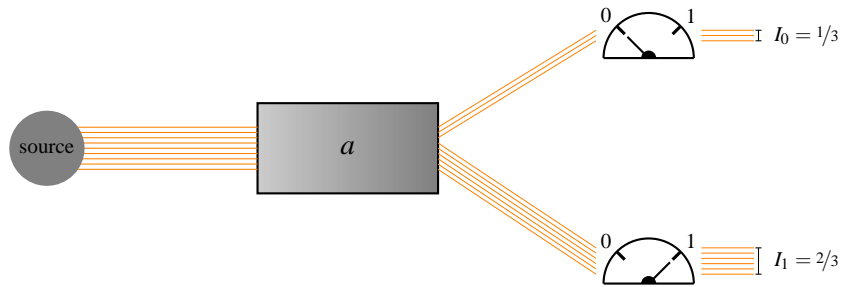

  \begin{center}
  \FigMeasInteractionBeam
  \end{center}
  \caption{Dichotomic measurement, $a$, performed on a \stress{beam} of particles, which may have been prepared in a large entangled state.
  This beam is split into two smaller beams, which hit each of the detectors corresponding to the possible outcomes, and whose intensities are recorded revealing the \stress{average behaviour} of a particle in the beam.}\label{fig:measinteractionbeam}
\end{figure}

In order to simplify the discussion, but with no crucial loss, we assume that the microscopic measurements are dichotomic,
i.e. $l = 2$, and take the possible outcomes to be $0$ and $1$. Then, the result of the macroscopic measurement (an intensity)
can be represented by a single number $I_1$ proportional to the number of particles that hit the detector corresponding to outcome $1$. 
(Note that this number can be normalised to yield a number in $[0,1]$ representing the proportion of particles that hit detector $1$.)

Of course, given the probabilistic nature of the microscopic measurements,
every time the whole experiment is run, with the same preparation of the initial state of the beam,
the number of particles hitting detector $1$ differs slightly.
But if $N$ is large enough, a realistic detector won't be able to discern, and so count, individual particles. 
For the purpose of this paper, 
we are concerned only with the mean, or expected, value of these intensities.
This can easily be obtained from such a macroscopic experiment, and it is what is usually taken to be the
\stress{value} of the macroscopic observable (e.g. total magnetisation).
This mean intensity can also be interpreted as giving the average behaviour among the particles in the beam or region:
if one would randomly select one of the $N$ particles and subject it to the microscopic measurement,
one would get the outcome $1$ with probability $I_1$ (assuming the intensities are normalised as mentioned above);
i.e. $I_1 = \frac{1}{N}\sum_{i=1}^N p_i(a = 1)$.
Observe that the situation is analogous to statistical mechanics, where
a macrostate arises as an averaging over an extremely large number of microstates,
and hence several different microstates can correspond to the same macrostate.

We, like Ramanathan et al. \cite{RamanathanEtAl:LocalRealismOfMacroscopicCorrelations},
are interested in average multipartite correlations arising from macroscopic measurements of this kind done across different sites.
We can think of a beam of photons being sent to each of a number of experimenters in spatially separated locations who can make several different measurements,
or we can think (as suggested in \cite{RamanathanEtAl:LocalRealismOfMacroscopicCorrelations})
of magnetisation measurements along several different directions done in a number of regions of a many-spin system.
In any case, we have a sort of average macroscopic Bell experiment:
the (mean) values of the macroscopic intensities (the intensity of an outcome $\tuple{x_1, \ldots, x_n}$ of a multipartite macroscopic measurement is the product of the intensities of the outcome $x_i$ for each site $i$) indicate the behaviour of a randomly chosen tuple of particles: one from each of the beams, or sites.
We shall show, as a consequence of Proposition \ref{prop:nkr-vorobev}, that,
as long as there are enough particles (microscopic sites) in each macroscopic site when compared to the number of possible measurement settings the experimenter at that site can choose from,
such average macroscopic behaviour is always local
no matter which no-signalling theory accounts for the 
underlying microscopic correlations.

We should mention how this relates to discussions about macroscopic locality in the literature.
In Bancal et al. \cite{BancalEtAl2008}, a bipartite setup similar to the one described above is considered.  
One difference is that the authors assume that the pair of beams received by the experimenters
is composed of independently and identically prepared pairs of particles.
This is also the case in Navascu\'es and Wunderlich \cite{NavascuesWunderlich2009}
(see the caption of figure 2 of this reference).
A run of their macroscopic experiment can be seen as
running the same microscopic bipartite Bell experiment multiple times and recording only how many times one obtains a certain outcome,
disregarding the information of which particles from each of the beams were originally paired. 
The mean values of the intensities (or, equivalently, the behaviour of a random pair) that we described above
are rather boring in this situation with identically prepared pairs: what we get is a diluted version of the probabilities for one of the identical microscopic empirical models\footnote{ Suppose that
$N$ pairs $\{\tuple{a^{(m)},b^{(m)}}\}_{m = 1, \dots, N}$ are prepared
each having the same empirical model,
described by probability distributions $p(a,b = x,y)$ where $a$ and $b$ range over
the possible measurement settings on each of the sites, and $x$ and $y$ over the respective outcomes.
Then the average behaviour of an (arbitrary) pair is given by:
\begin{align*}
  \tilde{p}(a,b = x,y) & = \frac{1}{N^2} \sum_{m_A,m_B = 1}^{N} p(a^{(m_A)},b^{(m_B)}=x,y)
\\ & = \frac{1}{N^2}\sum_{m=1}^Np(a^{(m)},b^{(m)}=x,y) + \frac{1}{N^2}\sum_{m_A \neq m_B = 1}^N p(a^{(m_A)},b^{(m_B)}=x,y)
\\ & = \frac{N}{N^2} p(a,b=x,y) + \frac{1}{N^2}\sum_{m_A \neq m_B = 1}^N p(a^{(m_A)} = x) p(b^{(m_B)} = y) 
\\ & = \frac{1}{N} p(a,b=x,y) + \frac{N^2 - N}{N^2}p(a = x) p(b = y) 
\\ & = \frac{1}{N} p(a,b=x,y) + \left(1-\frac{1}{N}\right)p(a = x) p(b = y) 
\end{align*}
which is the initial microscopic model very `diluted' by a local model
(corresponding to the pairs that were not prepared originally as a pair,
and are thus uncorrelated).
}.
However, these authors are not simply interested in the mean values of these intensities
(each of which is, after normalisation, a value in the interval $[0,1]$), but rather on the
more fundamental probability distributions over $[0,1]$ 
from which these means are calculated -- 
recall that each time the macroscopic experiment is run,
one measures slightly different intensities,
so the observed intensities fall in a distribution around the mean value.
Their aim is to witness non-locality on the fluctuations.
More specifically, they are concerned with the question of
whether these distributions of intensities
can be explained by local hidden variable models.

Even though some information is inevitably lost in such an experiment (particularly regarding the original pairings),
Bancal et al. \cite{BancalEtAl2008} show that one can witness non-locality at this level
if the detectors are perfect, in the sense that they can measure the intensity of beams
with maximal precision, to a sensitivity of one particle. This clearly becomes impractical as $N$ grows.
At the opposite end, in the idealised limit where the resolution of the detectors
is very bad, one would always observe the same intensities:
namely, our mean value of intensities with no fluctuations around it.
Navascu\'es \& Wunderlich \cite{NavascuesWunderlich2009}
suggest that it is physically reasonable that, when $N$ is large,
one could detect changes on intensity values of the order of $\sqrt N$.
These authors propose the notion of macroscopic locality to mean that
the distributions of observed intensities with a resolution of order $\sqrt N$ admit a local hidden variable explanation.
They show that this principle of macroscopic locality is satisfied by quantum mechanics,
but is not valid in general for all no-signalling theories: more accurately, the set of correlations
satisfying it is $Q^1$, the first level of the hierarchy of semidefinite programs
approximating the quantum set
proposed by Navascu\'es, Pironio, \& Ac\'in \cite{NavascuesPironioAcin2008}.

In this sense, the kind of macroscopic correlations we
(and Ramanathan et al. \cite{RamanathanEtAl:LocalRealismOfMacroscopicCorrelations}) consider
seems more restricted, since we show that these are local no matter which no-signalling theory accounts for the
underlying microscopic correlations.
However, there are some important differences, which we now summarise:
\begin{itemize}
  \item Firstly, we do not consider the beams to consist of
identically prepared pairs (or tuples) of particles.
In the bipartite setting of \cite{BancalEtAl2008,NavascuesWunderlich2009} above,
pairs of particles were identically and independently prepared and then a particle of each
pair was sent to Alice and the corresponding one to Bob (although the pairing
is lost as the particles are lumped together in a beam). In our setting,
the particles may be in different states and there are no restrictions on which groups of particles of Alice and of Bob (and possibly of others, as we allow for an arbitrary number of sites)
are entangled -- the `microstate' of the system can be very highly non-local.
The only restriction we impose
is that of no-signalling.
 \item Secondly,
our aim is not to explain the distribution of the intensities, with their fluctuations around the mean value,
by a local model, as in \cite{NavascuesWunderlich2009,BancalEtAl2008}. 
Rather, the (products of the) mean intensities themselves, which are taken as the value of the macroscopic observable,
give us a description of the behaviour of the average pair or tuple of particles in the beams. It is this average behaviour that we aim to explain by a local model. We prove this is indeed always possible for any no-signalling microscopic theory provided there are enough particles compared to measurement settings available at each site. The reason, again,  has to do with monogamy, which dilutes non-locality.
\end{itemize}

Despite the latter difference, note that
since such average behaviour corresponds to the mean, or expected, values of the
macroscopic measurements of intensities considered as in \cite{NavascuesWunderlich2009,BancalEtAl2008},
our result also implies
that macroscopic CHSH-type inequalities, i.e. inequalities involving only the expected values of macroscopic experiments,
can never be violated by no-signalling microscopic theories.
It is only by looking at higher-order moments (which correspond to other characteristics of the distribution, such as variance, skewness, kurtosis, etc.)
that one may witness a difference between quantum mechanics and general no-signalling theories.


\subsection{Macroscopic average behaviour: examples}\label{ssec:macro-avg-eg}
Let us see how this averaging works for our tripartite example.
We regard sites $B$ and $C$ as forming one macroscopic site, $M$,
and site $A$ as forming another\footnote{In this example, `macroscopic' means
one or two microsystems only, allowing us to keep the example small enough to be visualised.
This is sufficient to get local averages given that we are only considering two measurement settings per site:
recall from Section \ref{sec:introduction} that the condition
is that the number of microsystems in a site, or copies of a site, should be at least the number of
measurement settings available at that site
(except possibly for one of the sites, $A$ in this example, where we can consider a single microsystem or copy).}.
The idea is that we will average over the behaviour of the microsystems $B$ and $C$.
In order to be lumped together, $B$ and $C$ must be symmetric, i.e. of the same `type'.
In particular, we need to know which measurements on the site $B$ correspond to which measurements on the site $C$.
Here, we consider a symmetry of the system which makes the identifications
$b_1 \sim c_1$ and $b_2 \sim c_2$. We will name $m_1$ and $m_2$ the `macroscopic' measurements resulting from these identifications.

Given  an empirical model on the tripartite scenario,
one can consider the partial model on the subsystem composed of sites $A$ and $B$ only, whose probabilities are given
by marginalisation (in quantum mechanics, this corresponds to partial trace):
\[p(a_i,b_j = x, y ) := \sum_{z} p(a_i,b_j,c_k = x, y, z) \Mdot\]
Note that this expression is independent of $c_k$ due to no-signalling.
Similarly, one can consider the partial model on the subsystem composed of $A$ and $C$ only.

The average behaviour under the identification of $B$ with $C$ is then a bipartite model with two `macroscopic' sites $A$ and $M$,
given as an average of probability distributions of the partial models:
\begin{equation}\label{eq:macrodef-tripartite}
  p(a_i,m_j = x, y) := \frac{p(a_i,b_j = x,y) + p(a_i,c_j = x, y)}{2} \text{ .}
\end{equation}

Let us see how such average models look for the particular empirical models
\eqref{extable1}--\eqref{extable3} from Section \ref{ssec:empirical-models}.
The first two examples are symmetric with respect to $B$ and $C$, meaning that
the restriction of the model to sites $A$ and $B$ is the same as
the restriction of the model to sites $A$ and $C$.
Consequently, it is also equal to the macroscopic average model,
since the latter arises as an average of the two partial models.
The table for the macroscopic model emerging from example \eqref{extable1} is:
\[
\AMTableWstateZXmeasQUOT
\]
and for example \eqref{extable2} we obtain the totally mixed model:
\[
\AMTableTotalMix
\]
Note that both these (macroscopic average) bipartite models are local.
Now let us consider example \eqref{extable3}. The partial models on sites
$A$ and $B$ and on sites $A$ and $C$ are represented in the following tables:
\[
\ABTablePRmodel
\qquad\qquad
\ACTableTotalMix
\]
Note that the model in the left is non-local
(even maximally violating a Bell inequality: it is a Popescu--Rohrlich box \cite{PR-boxes}),
while the one in the right is local (it is the totally mixed model, in fact).
The  `macroscopic' average model is obtained as an average of these two:
\[
\AMTableOtherNonSymQUOT
\]
This model is also local, like the other average models above: a global probability distribution
for this model is
\begin{align*}
   &\tfrac{1}{8}[a_1a_2m_1m_2 = 0000]
  + \tfrac{1}{8}[a_1a_2m_1m_2 = 0001]
\\+&\tfrac{1}{8}[a_1a_2m_1m_2 = 0100]
  + \tfrac{1}{8}[a_1a_2m_1m_2 = 0110]
\\+&\tfrac{1}{8}[a_1a_2m_1m_2 = 1001]
  + \tfrac{1}{8}[a_1a_2m_1m_2 = 1011]
\\+&\tfrac{1}{8}[a_1a_2m_1m_2 = 1110]
  + \tfrac{1}{8}[a_1a_2m_1m_2 = 1111] \Mdot
\end{align*}
We shall see that these three examples are in no way special. Indeed,
our analysis will clarify that the macroscopic average behaviour is local
no matter which no-signalling tripartite empirical model we start from.

\subsection{Monogamy and macroscopic averages: Bell inequalities}\label{ssec:Bellineqs}
Now, we make a very simple observation that establishes the connection between monogamy of non-locality and
locality of these macroscopic averages.
Consider any Bell inequality $\mathcal{B}(-,-) \leq R$ for a scenario with two parties,
each with two available measurements.
Such an inequality is determined by a set of coefficients $\alpha(i,j,x,y)$ and a bound $R$. We have that:
\begin{calculation}
\mathcal{B}(A,M) \leq R
\ejust\bimplies
\sum_{i,j,x,y}\alpha(i,j,x,y)p(a_i,m_j=x,y) \leq R
\just\bimplies{definition of the average probabilities, eq. \ref{eq:macrodef-tripartite}}
 \sum_{i,j,x,y}\alpha(i,j,x,y)\;\frac{p(a_i,b_j=x,y)+p(a_i,c_j=x,y)}{2} \leq R
 \just\bimplies{re-arranging terms}
\sum_{i,j,x,y}\alpha(i,j,x,y)p(a_i,b_j=x,y)  +  \sum_{i,j,x,y}\alpha(i,j,x,y)p(a_i,c_j=x,y) \leq 2 R
\ejust\bimplies
\mathcal{B}(A,B) + \mathcal{B}(A,C) \leq 2 R
\end{calculation}
That is, the `macroscopic' average model, $p(a_i,m_j=\cdots)$ on sites $A$ and $M$,
satisfies the Bell inequality, $\mathcal{B}(A,M) \leq R$,  if and only if
the `microscopic' model (on sites $A$, $B$ and $C$) is monogamous with respect to violating it;
i.e. the bipartite
partial models $p(a_i,b_j=\cdots)$ and $p(a_i,c_j=\cdots)$
satisfy the monogamy relation $\mathcal{B}(A,B) + \mathcal{B}(A,C) \leq 2 R$,
and so cannot both violate the Bell inequality.
This is an instance of a more general equivalence between Bell inequalities on `macroscopic' averages
and the monogamy of violation of the same inequality at the `microscopic' level. 
As a consequence, a macroscopic average model satisfies all Bell inequalities (i.e. it is local) if and only if the microscopic model is monogamous with respect to violating all those inequalities. This is the case, in particular, of all the models in the tripartite scenario we are analysing, such as the examples considered in Section \ref{ssec:macro-avg-eg}.
In the next section, we give a reason for this based on the structure of the scenario.

\section{A structural explanation}\label{sec:structuralexplanation}
\subsection{\Vorobev's theorem}
A classical mathematical result due to \Vorobev~\cite{vorobev}, and motivated by a problem in game theory,
deals with the following question, here rephrased in our terms:
for which measurement scenarios $\UU$ (or $\Sigma$) is it so
that any no-signalling empirical model 
$(\mu_C)_{C \in \UU}$ defined on it
admits a global extension, i.e. is local or non-contextual?
\Vorobev\ derived a necessary and sufficient condition on the simplicial complex $\Sigma$
for this to be the case.
We present a simplified yet equivalent version of \Vorobev's condition,
which happens to be known in relational database theory as acyclicity, an important property of database schemata
(cf. \cite{Abramsky12:databases,RSB:DPhil-thesis-forth} for more on the connection between relational database theory and the study of locality and non-contextuality).
The idea is that such a scenario can be constructed by adding one measurement at a time
in such a way that the new measurement is added to only one maximal context.
Equivalently, it can be de-constructed by removing at each step a measurement belonging to a single maximal context.

\begin{definition}\label{def:acyclicity}
  Let $\Sigma$ be a simplicial complex.
  Given a maximal face $C$, let $\proper{C}$ denote the vertices of $\Sigma$ which belong to $C$ and not to any other maximal faces.
  If $\proper{C} \neq \emptyset$ for some $C$, then we say that there is a \emph{Graham-reduction} step from $\Sigma$ to the subcomplex
  \[\Sigma' :=  \setdef{\sigma \in \Sigma}{\sigma \cap \proper{C} = \emptyset} = \setdef{\sigma \,\setminus\, \proper{C}}{\sigma \in \Sigma}\]
  and write $\Sigma \leadsto \Sigma'$.

  The complex $\Sigma$ is said to be \emph{acyclic} if
  it is Graham-reducible to the empty complex\footnote{The empty complex, $\mathbf{0}$ is the only simplicial complex on $\emptyset$, that is, with no vertices.},
  i.e.
  if there exists a series of Graham-reduction steps from $\Sigma$ to the empty complex:
  \[\Sigma = \Sigma_0 \leadsto \Sigma_1 \leadsto \cdots \leadsto \Sigma_r = \mathbf{0} \Mdot\]
 \end{definition}

The following is an example of a successful Graham reduction to $\mathbf{0}$, witnessing the acyclicity of the
simplicial complex on the left.
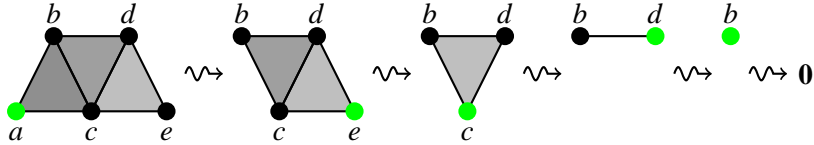
\begin{center}
  \begin{tikzpicture}[scale=1]
  \path (-1 , 0 ) coordinate (A);
  \path (-.5, 1 ) coordinate (B);
  \path ( 0 , 0 ) coordinate (C);
  \path (+.5, 1 ) coordinate (D);
  \path (+1 , 0 ) coordinate (E);

  \begin{scope}[opacity=0.5, color=gray]
    \fill (A) -- (B) -- (C) -- cycle;
    \fill (B) -- (C) -- (D) -- cycle;
    \fill (C) -- (D) -- (E) -- cycle;
  \end{scope}

  \begin{scope}[ thick,black]
    \draw (A) -- (B) -- (C) -- cycle;
    \draw (B) -- (C) -- (D) -- cycle;
    \draw (C) -- (D) -- (E) -- cycle;
  \end{scope}
 
  \fill[green]    (A) circle(.12) ;
  \fill[black]  (B) circle(.12) ;
  \fill[black]  (C) circle(.12) ;
  \fill[black]  (D) circle(.12) ;
  \fill[black]  (E) circle(.12) ;
  
  \path (A) +(0,-.3) node {$a$}; 
  \path (B) +(0,+.3) node {$b$};
  \path (C) +(0,-.3) node {$c$};
  \path (D) +(0,+.3) node {$d$};
  \path (E) +(0,-.3) node {$e$};


  \path (2, 1 ) coordinate (B);
  \path ( 2.5 , 0 ) coordinate (C);
  \path (3, 1 ) coordinate (D);
  \path ( 3.5 , 0 ) coordinate (E);

  \begin{scope}[opacity=0.5, color=gray]
    \fill (B) -- (C) -- (D) -- cycle;
    \fill (C) -- (D) -- (E) -- cycle;
  \end{scope}

  \begin{scope}[ thick,black]
    \draw (B) -- (C) -- (D) -- cycle;
    \draw (C) -- (D) -- (E) -- cycle;
  \end{scope}
 
  \fill[black]  (B) circle(.12) ;
  \fill[black]  (C) circle(.12) ;
  \fill[black]  (D) circle(.12) ;
  \fill[green]    (E) circle(.12) ;
  
  \path (B) +(0,+.3) node {$b$};
  \path (C) +(0,-.3) node {$c$};
  \path (D) +(0,+.3) node {$d$};
  \path (E) +(0,-.3) node {$e$};


  \path (4.5, 1 ) coordinate (B);
  \path ( 5 , 0 ) coordinate (C);
  \path (5.5, 1 ) coordinate (D);

  \begin{scope}[opacity=0.5, color=gray]
    \fill (B) -- (C) -- (D) -- cycle;
  \end{scope}

  \begin{scope}[ thick,black]
    \draw (B) -- (C) -- (D) -- cycle;
  \end{scope}
 
  \fill[black]  (B) circle(.12) ;
  \fill[green]  (C) circle(.12) ;
  \fill[black]  (D) circle(.12) ;
  

  \path (B) +(0,+.3) node {$b$};
  \path (C) +(0,-.3) node {$c$};
  \path (D) +(0,+.3) node {$d$};

  \path (6.5, 1 ) coordinate (B);
  \path (7.5, 1 ) coordinate (D);

  \begin{scope}[ thick,black]
    \draw (B) -- (D);
  \end{scope}
 
  \fill[black]  (B) circle(.12) ;
  \fill[green]  (D) circle(.12) ;
  
  \path (B) +(0,+.3) node {$b$};
  \path (D) +(0,+.3) node {$d$};


  \path (8.5, 1 ) coordinate (B);

  \fill[green]  (B) circle(.12) ;
  
  \path (B) +(0,+.3) node {$b$};


  \path (9.5, .5) node{$\mathbf{0}$};

\draw [thick,->,snake=coil,segment length=6pt,segment aspect=0] (1.25,.5) -- (1.75,.5);
\draw [thick,->,snake=coil,segment length=6pt,segment aspect=0] (3.75,.5) -- (4.25,.5);
\draw [thick,->,snake=coil,segment length=6pt,segment aspect=0] (5.75,.5) -- (6.25,.5);
\draw [thick,->,snake=coil,segment length=6pt,segment aspect=0] (7.75,.5) -- (8.25,.5);
\draw [thick,->,snake=coil,segment length=6pt,segment aspect=0] (8.75,.5) -- (9.25,.5);
\end{tikzpicture}
\end{center}
On the contrary, the following complex is not acyclic: Graham reduction always fails, hitting a `cycle'.
\begin{center}
\begin{tikzpicture}[scale=1]
  \path (-1 , 0 ) coordinate (A);
  \path (-.5, 1 ) coordinate (B);
  \path ( 0 , 0 ) coordinate (C);
  \path (+.5, 1 ) coordinate (D);
  \path (+1 , 0 ) coordinate (E);

  \begin{scope}[opacity=0.5, color=gray]
    \fill (A) -- (B) -- (C) -- cycle;
    \fill (C) -- (D) -- (E) -- cycle;
  \end{scope}

  \begin{scope}[ thick,black]
    \draw (A) -- (B) -- (C) -- cycle;
    \draw (B) -- (C) -- (D) -- cycle;
    \draw (C) -- (D) -- (E) -- cycle;
  \end{scope}
 
  \fill[green]    (A) circle(.12) ;
  \fill[black]  (B) circle(.12) ;
  \fill[black]  (C) circle(.12) ;
  \fill[black]  (D) circle(.12) ;
  \fill[black]  (E) circle(.12) ;
  
  \path (A) +(0,-.3) node {$a$}; 
  \path (B) +(0,+.3) node {$b$};
  \path (C) +(0,-.3) node {$c$};
  \path (D) +(0,+.3) node {$d$};
  \path (E) +(0,-.3) node {$e$};


  \path (2, 1 ) coordinate (B);
  \path ( 2.5 , 0 ) coordinate (C);
  \path (3, 1 ) coordinate (D);
  \path ( 3.5 , 0 ) coordinate (E);

  \begin{scope}[opacity=0.5, color=gray]
    \fill (C) -- (D) -- (E) -- cycle;
  \end{scope}

  \begin{scope}[ thick,black]
    \draw (B) -- (C) -- (D) -- cycle;
    \draw (C) -- (D) -- (E) -- cycle;
  \end{scope}
 
  \fill[black]  (B) circle(.12) ;
  \fill[black]  (C) circle(.12) ;
  \fill[black]  (D) circle(.12) ;
  \fill[green]    (E) circle(.12) ;
  
  \path (B) +(0,+.3) node {$b$};
  \path (C) +(0,-.3) node {$c$};
  \path (D) +(0,+.3) node {$d$};
  \path (E) +(0,-.3) node {$e$};


  \path (4.5, 1 ) coordinate (B);
  \path ( 5 , 0 ) coordinate (C);
  \path (5.5, 1 ) coordinate (D);

  \begin{scope}[ thick,black]
    \draw (B) -- (C) -- (D) -- cycle;
  \end{scope}

  \draw[ thick, black] (B) -- (D);
  \draw[ thick, black] (C) -- (D);
  \draw[ thick, black] (C) -- (B);

  \fill[black]  (B) circle(.12) ;
  \fill[black]  (C) circle(.12) ;
  \fill[black]  (D) circle(.12) ;
  

%
%
%
%
%
%
%
%
%
%
%


  \fill[red] (6.5, .9) -- (6.6, 1) -- (7.5,0.1) -- (7.4,0) -- cycle;
  \fill[red] (6.5, .1) -- (6.6, 0) -- (7.5,0.9) -- (7.4,1) -- cycle;

\draw [thick,->,snake=coil,segment length=6pt,segment aspect=0] (1.25,.5) -- (1.75,.5);
\draw [thick,->,snake=coil,segment length=6pt,segment aspect=0] (3.75,.5) -- (4.25,.5);
\draw [thick,->,snake=coil,segment length=6pt,segment aspect=0] (5.75,.5) -- (6.25,.5);
  
\end{tikzpicture}
\end{center}

\begin{theorem}[\Vorobev~\cite{vorobev}, rephrased and with simplified condition \cite{RSB:DPhil-thesis-forth}]
  Let $\Sigma$ be a simplicial complex.
  Then any empirical model defined on $\Sigma$ is extendable
  if and only if $\Sigma$ is acyclic.
\end{theorem}

\subsection{Structural reason: tripartite example}\label{ssec:structural-tripartiteexample}
We mentioned above that, for the scenario we are considering,
any empirical model gives rise to local average behaviour correlations.
The structural reason for this is the fact that the quotient of the scenario by the identification
of sites $B$ and $C$ is acyclic. 
Let us look at this in more detail.

Our scenario is represented by the simplicial complex
$\Disj_2 \sjoin \Disj_2 \sjoin \Disj_2$, where the factors correspond to sites $A$, $B$ and $C$.
This is the hollow octahedron we depicted before, in Section \ref{ssec:empirical-models}.
Given that we want to identify $B$ and $C$,
we regard this complex as\footnote{\label{footnote:notationsigma}The $\Sigma_{n,\vec{k},\vec{r}}$ notation on the left-hand side
will be introduced in Section \ref{sec:gen-multipartite}; it is provided here just for reference:
$n$ stands for the number of `macroscopic' sites,
$k_i$ for the number of measurement settings available at site $i$, and $r_i$ for the number of `microscopic' sites in, or copies of, site $i$.
The reader is referred to Section \ref{ssec:meas-scenarios} for the notation on the right-hand side.}
\[\Sigma_{n=2,k_1=2,k_2=2,r_1=1,r_2=2} \;\;\;\;:=\;\;\;\; \Disj_2 \sjoin \Disj_2^{\sjoin 2} \;\;=\;\; \Disj_2 \sjoin (\Disj_2 \sjoin \Disj_2) \Mcomma\]
with sites $B$ and $C$ `grouped' together in the second factor on which the identification $b_i \sim c_i$ acts.

We shall explicitly see what the quotient is.
The first step is so-called semiregularisation, where we remove edges between vertices that are being identified
as such edges are unnecessary.
So, in this case,
we must remove the edges $\enset{b_1,c_1}$ and $\enset{b_2,c_2}$, obtaining the following simplicial complex:
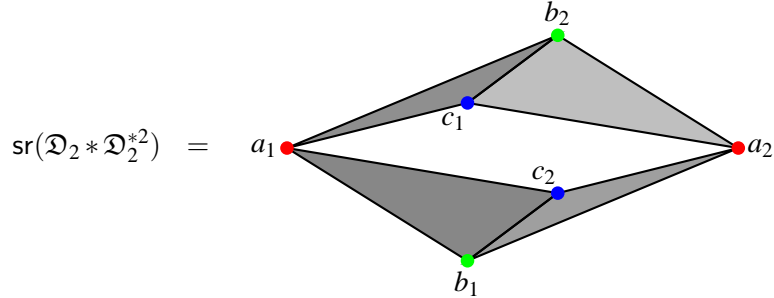
\[
  \sr(\Disj_2 \sjoin \Disj_2^{\sjoin 2}) \;\;\;=\;\;\;\; 
\begin{tikzpicture}[baseline={([yshift=-.5ex]current bounding box.center)},scale=3]
  \path (-1 , 0 ) coordinate (A1);
  \path ( 1 , 0 ) coordinate (A2);
  \path (-.2,-.5) coordinate (B1);
  \path (+.2,+.5) coordinate (B2);
  \path (-.2,+.2) coordinate (C1);
  \path (+.2,-.2) coordinate (C2);

  \begin{scope}[opacity=0.5, color=gray]
    \fill (A1) -- (B1) -- (C2) -- cycle;
    \fill (A1) -- (B2) -- (C1) -- cycle;
    \fill (A2) -- (B1) -- (C2) -- cycle;
    \fill (A2) -- (B2) -- (C1) -- cycle;
  \end{scope}

  \begin{scope}[ thick,black]
    \draw (A1) -- (B1) -- (C2) -- cycle;
    \draw (A1) -- (B2) -- (C1) -- cycle;
    \draw (A2) -- (B1) -- (C2) -- cycle;
    \draw (A2) -- (B2) -- (C1) -- cycle;
  \end{scope}
 
  \fill[red]   (A1) circle(.03) ;
  \fill[red]   (A2) circle(.03) ;
  \fill[green] (B1) circle(.03) ;
  \fill[green] (B2) circle(.03) ;
  \fill[blue]  (C1) circle(.03) ;
  \fill[blue]  (C2) circle(.03) ;
  
  \path (A1) +(-.1 ,0  ) node {$a_1$}; 
  \path (A2) +(+.1 ,0  ) node {$a_2$};
  \path (B1) +(0   ,-.1) node {$b_1$};
  \path (B2) +(0   ,+.1) node {$b_2$};
  \path (C1) +(-.06,-.08) node {$c_1$};
  \path (C2) +(-.06,+.08) node {$c_2$};
\end{tikzpicture}
\]
Now, taking the quotient, we will identify the measurements $b_1$ and $c_1$ as $m_1$, and
$b_2$ and $c_2$ as $m_2$. We obtain the following simplicial complex:
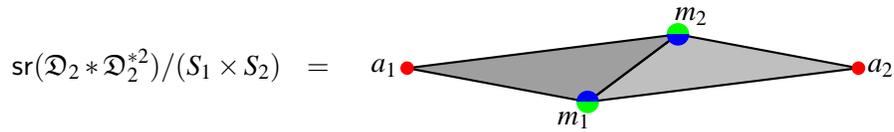
\[
  \sr(\Disj_2 \sjoin \Disj_2^{\sjoin 2})/(S_1 \times S_2) \;\;\;=\;\;\;\; 
\begin{tikzpicture}[baseline={([yshift=-.5ex]current bounding box.center)},scale=3]
  \path (-1 , 0 ) coordinate (A1);
  \path ( 1 , 0 ) coordinate (A2);
  \path (-.2,-.15) coordinate (M1);
  \path (+.2,+.15) coordinate (M2);
  \path (-.2,-.5) coordinate (B1);
  \path (+.2,+.5) coordinate (B2);
  \path (-.2,+.2) coordinate (C1);
  \path (+.2,-.2) coordinate (C2);

  \begin{scope}[opacity=0.5, color=gray]
    \fill (A1) -- (M1) -- (M2) -- cycle;
    \fill (A2) -- (M1) -- (M2) -- cycle;
  \end{scope}

  \begin{scope}[ thick,black]
    \draw (A1) -- (M1) -- (M2) -- cycle;
    \draw (A2) -- (M1) -- (M2) -- cycle;
  \end{scope}
 
  \fill[red]   (A1) circle(.03) ;
  \fill[red]   (A2) circle(.03) ;

  \begin{scope}
  \clip (M1) circle (.05);
  \fill[blue]  (M1) rectangle +(-1,+1);
  \fill[blue]  (M1) rectangle +(+1,+1);
  \fill[green] (M1) rectangle +(-1,-1);
  \fill[green] (M1) rectangle +(+1,-1);
  \end{scope}
  \begin{scope}
  \clip (M2) circle (.05);
  \fill[blue]  (M2) rectangle +(-1,-1);
  \fill[blue]  (M2) rectangle +(+1,-1);
  \fill[green] (M2) rectangle +(-1,+1);
  \fill[green] (M2) rectangle +(+1,+1);
  \end{scope}
  
  \path (A1) +(-.1 ,0  ) node {$a_1$}; 
  \path (A2) +(+.1 ,0  ) node {$a_2$};
  \path (M1) +(-.06,-.08) node {$m_1$};
  \path (M2) +(+.06,+.08) node {$m_2$};
\end{tikzpicture}
\]
Observe that the set of maximal faces (i.e. the cover of maximal contexts) is
\[\enset{\enset{a_1,m_1,m_2},\enset{a_2,m_1,m_2}} \Mdot\]
So, more things are compatible than in the usual bipartite scenario, which has cover
\[\enset{\enset{a_1,m_1},\enset{a_1,m_2},\enset{a_2,m_1},\enset{a_2,m_2}} \Mdot\]

As it happens, any empirical model defined on the original complex $\Disj_2 \sjoin \Disj_2^2$
will give rise to another model defined on the quotient scenario, by taking averages along the faces
being identified. Therefore, not only are the probabilities $p(a_i,m_j=\dots)$ defined via an average,
giving a model on the usual bipartite scenario above,
so are the probabilities $p(a_i, m_1, m2 = \dots)$,
yielding a model on the more compatible bipartite scenario that arises as a quotient.
The probability distribution on the triangle $\enset{a_i,m_1,m_2}$
is obtained as an average of the probability distributions on the top and bottom triangles that gave rise to it, namely
of $p(a_i,b_1,c_2 = \dots)$ and $p(a_i,c_1,b_2 = \dots)$.

The quotient complex we have obtained does satisfy the \Vorobev~ condition of acyclicity.
This is easy to see: one can remove the vertices, for example, in the order $a_1$, $a_2$, $m_1$, $m_2$.
Therefore,  no matter
which empirical model $p(a_i,b_j,c_k=\cdots)$ we start from, the model of average macroscopic behaviour, $p(a_i,m_j=\cdots)$,
is local. In particular, it satisfies any  Bell inequality.
Hence, by the equivalence discussed in Section \ref{ssec:Bellineqs}, the original
tripartite model also satisfies a monogamy relation for any of these bipartite Bell inequalities.

\subsection{A non-acyclic example}\label{ssec:structural-nonacyclicexample}
Let us now consider an example where one does not get monogamy relations,
or equivalently, where one does not necessarily get local macroscopic averages. 
Suppose that we again have a tripartite ($A$, $B$, $C$) scenario, but that this time
$B$ and $C$ have 3 available measurement settings each.
In a compositional notation (as explained in footnote \ref{footnote:notationsigma}
at the start of Section \ref{ssec:structural-tripartiteexample})
and since we are again interested in identifying the sites $B$ and $C$,
this scenario is represented by the simplicial complex
\[\Sigma_{n=2,k_1=2,k_2=3,r_1=1,r_2=2} \;\;\;\;:=\;\;\;\; \Disj_2 \sjoin \Disj_3^{\sjoin 2} \;\;=\;\; \Disj_2 \sjoin (\Disj_3 * \Disj_3) \Mdot\]
The maximal contexts are
\begin{align*}
\UU = \{ &
\enset{a_1,b_1,c_1},
\enset{a_1,b_1,c_2},
\enset{a_1,b_1,c_3},
\enset{a_1,b_2,c_1},
\enset{a_1,b_2,c_2},
\enset{a_1,b_2,c_3},
\\&
\enset{a_1,b_3,c_1},
\enset{a_1,b_3,c_2},
\enset{a_1,b_3,c_3},
\enset{a_2,b_1,c_1},
\enset{a_2,b_1,c_2},
\enset{a_2,b_1,c_3},
\\&
\enset{a_2,b_2,c_1},
\enset{a_2,b_2,c_2},
\enset{a_2,b_2,c_3},
\enset{a_2,b_3,c_1},
\enset{a_2,b_3,c_2},
\enset{a_2,b_3,c_3} \}
\end{align*}
and half the simplicial complex is depicted below
(one should imagine the other half, a mirror image of this consisting of
the faces that include $a_2$ instead of $a_1$, collated to it;
we choose  to omit that part as it would make the picture more confusing and hard to visualise):
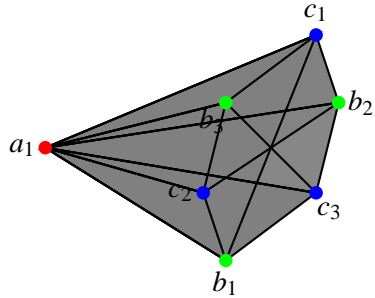
\begin{center}
\begin{tikzpicture}[scale=3]
  \path (-1 , 0 )  coordinate (A1);
  \path (-.2,-.5)  coordinate (B1);
  \path (+.2,+.5)  coordinate (C1);
  \path (-.3,-.2) coordinate (C2);
  \path (-.2,+.2)  coordinate (B3);
  \path (+.2,-.2)  coordinate (C3);
  \path (+.3,+.2) coordinate (B2);

  \begin{scope}[opacity=0.5, color=gray]
    \fill (A1) -- (B1) -- (C1) -- cycle;
    \fill (A1) -- (B1) -- (C2) -- cycle;
    \fill (A1) -- (B1) -- (C3) -- cycle;
    \fill (A1) -- (B2) -- (C1) -- cycle;
    \fill (A1) -- (B2) -- (C2) -- cycle;
    \fill (A1) -- (B2) -- (C3) -- cycle;
    \fill (A1) -- (B3) -- (C1) -- cycle;
    \fill (A1) -- (B3) -- (C2) -- cycle;
    \fill (A1) -- (B3) -- (C3) -- cycle;
  \end{scope}

  \begin{scope}[ thick,black]
    \draw (A1) -- (B1) -- (C1) -- cycle;
    \draw (A1) -- (B1) -- (C2) -- cycle;
    \draw (A1) -- (B1) -- (C3) -- cycle;
    \draw (A1) -- (B2) -- (C1) -- cycle;
    \draw (A1) -- (B2) -- (C2) -- cycle;
    \draw (A1) -- (B2) -- (C3) -- cycle;
    \draw (A1) -- (B3) -- (C1) -- cycle;
    \draw (A1) -- (B3) -- (C2) -- cycle;
    \draw (A1) -- (B3) -- (C3) -- cycle;
  \end{scope}
 
  \fill[red]   (A1) circle(.03) ;
  \fill[green] (B1) circle(.03) ;
  \fill[green] (B2) circle(.03) ;
  \fill[green] (B3) circle(.03) ;
  \fill[blue]  (C1) circle(.03) ;
  \fill[blue]  (C2) circle(.03) ;
  \fill[blue]  (C3) circle(.03) ;
  
  \path (A1) +(-.1 ,0  )  node {$a_1$}; 
  \path (B1) +(0   ,-.1)  node {$b_1$};
  \path (B2) +(+.1   ,0)  node {$b_2$};
  \path (B3) +(-.06,-.08) node {$b_3$};
  \path (C1) +(0   ,+.1)  node {$c_1$};
  \path (C2) +(-.1 ,  0)  node {$c_2$};
  \path (C3) +(+.06,-.08) node {$c_3$};
\end{tikzpicture}
\end{center}
We consider the identifications $b_i \sim c_i$ ($i=1,2,3$).
Again, we first discard the edges between identified measurements, namely $\enset{b_1,c_1}$, $\enset{b_2,c_2}$, and $\enset{b_3,c_3}$. 
The resulting complex, $\sr(\Disj_2 \sjoin \Disj_3^{\sjoin 2})$, is depicted
below (as above, we just depict half of it):
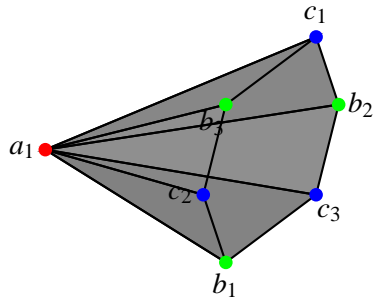
\begin{center}
\begin{tikzpicture}[scale=3]
  \path (-1 , 0 )  coordinate (A1);
  \path (-.2,-.5)  coordinate (B1);
  \path (+.2,+.5)  coordinate (C1);
  \path (-.3,-.2) coordinate (C2);
  \path (-.2,+.2)  coordinate (B3);
  \path (+.2,-.2)  coordinate (C3);
  \path (+.3,+.2) coordinate (B2);

  \begin{scope}[opacity=0.5, color=gray]
    \fill (A1) -- (B1) -- (C2) -- cycle;
    \fill (A1) -- (B1) -- (C3) -- cycle;
    \fill (A1) -- (B2) -- (C1) -- cycle;
    \fill (A1) -- (B2) -- (C3) -- cycle;
    \fill (A1) -- (B3) -- (C1) -- cycle;
    \fill (A1) -- (B3) -- (C2) -- cycle;
   \end{scope}

  \begin{scope}[ thick,black]
    \draw (A1) -- (B1) -- (C2) -- cycle;
    \draw (A1) -- (B1) -- (C3) -- cycle;
    \draw (A1) -- (B2) -- (C1) -- cycle;
    \draw (A1) -- (B2) -- (C3) -- cycle;
    \draw (A1) -- (B3) -- (C1) -- cycle;
    \draw (A1) -- (B3) -- (C2) -- cycle;
  \end{scope}
 
  \fill[red]   (A1) circle(.03) ;
  \fill[green] (B1) circle(.03) ;
  \fill[green] (B2) circle(.03) ;
  \fill[green] (B3) circle(.03) ;
  \fill[blue]  (C1) circle(.03) ;
  \fill[blue]  (C2) circle(.03) ;
  \fill[blue]  (C3) circle(.03) ;
  
  \path (A1) +(-.1 ,0  )  node {$a_1$}; 
  \path (B1) +(0   ,-.1)  node {$b_1$};
  \path (B2) +(+.1   ,0)  node {$b_2$};
  \path (B3) +(-.06,-.08) node {$b_3$};
  \path (C1) +(0   ,+.1)  node {$c_1$};
  \path (C2) +(-.1 ,  0)  node {$c_2$};
  \path (C3) +(+.06,-.08) node {$c_3$};
\end{tikzpicture}
\end{center}
The quotient then identifies $b_i$ with $c_i$, yielding the following
simplicial complex (half of it, as before):
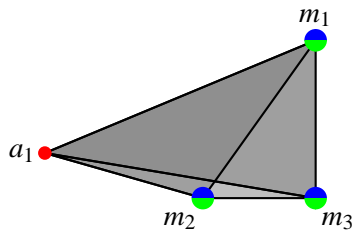
\begin{center}
\begin{tikzpicture}[scale=3]
  \path (-1 , 0 )  coordinate (A1);
  \path (-.2,-.5)  coordinate (B1);
  \path (+.2,+.5)  coordinate (M1);
  \path (-.3,-.2)  coordinate (M2);
  \path (-.2,+.2)  coordinate (B3);
  \path (+.2,-.2)  coordinate (M3);
  \path (+.3,+.2)  coordinate (B2);

  \begin{scope}[opacity=0.5, color=gray]
    \fill (A1) -- (M1) -- (M2) -- cycle;
    \fill (A1) -- (M1) -- (M3) -- cycle;
    \fill (A1) -- (M2) -- (M3) -- cycle;
  \end{scope}

  \begin{scope}[ thick,black]
    \draw (A1) -- (M1) -- (M2) -- cycle;
    \draw (A1) -- (M1) -- (M3) -- cycle;
    \draw (A1) -- (M2) -- (M3) -- cycle;
  \end{scope}
 
  \fill[red]   (A1) circle(.03) ;
 \begin{scope}
  \clip (M1) circle (.05);
  \fill[blue]  (M1) rectangle +(-1,+1);
  \fill[blue]  (M1) rectangle +(+1,+1);
  \fill[green] (M1) rectangle +(-1,-1);
  \fill[green] (M1) rectangle +(+1,-1);
  \end{scope}
  \begin{scope}
  \clip (M2) circle (.05);
  \fill[blue]  (M2) rectangle +(-1,+1);
  \fill[blue]  (M2) rectangle +(+1,+1);
  \fill[green] (M2) rectangle +(-1,-1);
  \fill[green] (M2) rectangle +(+1,-1);
  \end{scope} 
  \begin{scope}
  \clip (M3) circle (.05);
  \fill[blue]  (M3) rectangle +(-1,+1);
  \fill[blue]  (M3) rectangle +(+1,+1);
  \fill[green] (M3) rectangle +(-1,-1);
  \fill[green] (M3) rectangle +(+1,-1);
  \end{scope}
  
  \path (A1) +(-.1 ,0  )  node {$a_1$}; 
  \path (M1) +(0   ,+.1)  node {$m_1$};
  \path (M2) +(-.1,-.1) node {$m_2$};
  \path (M3) +(+.1,-.1) node {$m_3$};
\end{tikzpicture}
\end{center}
Collating the missing half of the picture, this is a hollow triangular bipyramid,
a complex with six two-dimensional maximal faces:
\begin{align*}
       \{ & \enset{a_1,m_1,m_2}, \enset{a_1,m_2,m_3}, \enset{a_1,m_3,m_1}, \\
          & \enset{a_2,m_1,m_2}, \enset{a_2,m_2,m_3}, \enset{a_2,m_3,m_1}\} 
\end{align*}
which clearly does not satisfy the acyclicity condition of \Vorobev's theorem.
Indeed, one can find empirical models
for the original measurement scenario whose `quotient' average macroscopic behaviour
is non-local. The reason for this is that we have too many measurement settings available
and not enough microscopic sites (or independent copies of Bob) to dilute the information these measurements can obtain.

The situation becomes different if there is another site $D$ (with measurements
$d_1,d_2,d_3$) and the sites $B$, $C$ and $D$ are identified as forming the same macroscopic site (or as being three copies of Bob).
The complex in this case is 
\[\Sigma_{n=2,k_1=2,k_2=3,r_1=1,r_2=3} \;\;\;\;:=\;\;\;\; \Disj_2 \sjoin \Disj_3^{\sjoin 3} \;\;=\;\; \Disj_2 \sjoin (\Disj_3 \sjoin \Disj_3 \sjoin \Disj_3) \Mcomma\]
and its quotient is a solid, rather than hollow, triangular bipyramid (two filled tetrahedrons collated together).
In keeping with the previous examples, we depict only half of the simplicial complex: 
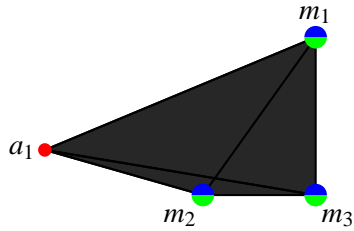
\begin{center}
\begin{tikzpicture}[scale=3]
  \path (-1 , 0 )  coordinate (A1);
  \path (-.2,-.5)  coordinate (B1);
  \path (+.2,+.5)  coordinate (M1);
  \path (-.3,-.2)  coordinate (M2);
  \path (-.2,+.2)  coordinate (B3);
  \path (+.2,-.2)  coordinate (M3);
  \path (+.3,+.2)  coordinate (B2);

  \begin{scope}[opacity=0.85]
    \fill (A1) -- (M1) -- (M3) -- (M2) -- cycle;
  \end{scope}

  \begin{scope}[ thick,black]
    \draw (A1) -- (M1) -- (M2) -- cycle;
    \draw (A1) -- (M1) -- (M3) -- cycle;
    \draw (A1) -- (M2) -- (M3) -- cycle;
  \end{scope}
 
  \fill[red]   (A1) circle(.03) ;
 \begin{scope}
  \clip (M1) circle (.05);
  \fill[blue]  (M1) rectangle +(-1,+1);
  \fill[blue]  (M1) rectangle +(+1,+1);
  \fill[green] (M1) rectangle +(-1,-1);
  \fill[green] (M1) rectangle +(+1,-1);
  \end{scope}
  \begin{scope}
  \clip (M2) circle (.05);
  \fill[blue]  (M2) rectangle +(-1,+1);
  \fill[blue]  (M2) rectangle +(+1,+1);
  \fill[green] (M2) rectangle +(-1,-1);
  \fill[green] (M2) rectangle +(+1,-1);
  \end{scope} 
  \begin{scope}
  \clip (M3) circle (.05);
  \fill[blue]  (M3) rectangle +(-1,+1);
  \fill[blue]  (M3) rectangle +(+1,+1);
  \fill[green] (M3) rectangle +(-1,-1);
  \fill[green] (M3) rectangle +(+1,-1);
  \end{scope}
  
  \path (A1) +(-.1 ,0  )  node {$a_1$}; 
  \path (M1) +(0   ,+.1)  node {$m_1$};
  \path (M2) +(-.1,-.1) node {$m_2$};
  \path (M3) +(+.1,-.1) node {$m_3$};
\end{tikzpicture}
\end{center}
The set of maximal faces of this complex (i.e. the cover of maximal contexts of this scenario) is
\[ \enset{\enset{a_1,m_1,m_2,m_3}, \enset{a_2,m_1,m_2,m_3}} \Mcomma\]
from which it is clear that the complex is acyclic,
ensuring that any no-signalling model satisfies all monogamy relations, and that all average macroscopic
models are local.
The point we are hinting at is that, in order to guarantee monogamy and local averages,
there must be at least as many microscopic sites in each macroscopic site as there are
measurement settings available at that site.

\section{General multipartite scenarios}\label{sec:gen-multipartite}
We now look at multipartite scenarios in general.
We consider the general scenario already mentioned in the item list in Section \ref{sec:introduction}:
we have $n$ (macroscopic) sites $1, \ldots, n$ (also denoted by $A, B, C, \ldots$);
each site $i$ has $k_i$ measurement settings;
and we have $r_i$ copies of site $i$, or microscopic sites constituting the macroscopic site $i$.
If we write $A$ for a (macroscopic) site,
then 
$A^{(1)}, \ldots, A^{(r_A)}$
denote the several copies of it or microscopic sites constituting it,
and
$a^{(m)}_1, \dots, a^{(m)}_{k_A}$
are the measurements for the $m$-th copy or microscopic site $A^{(m)}$, where $m \in \enset{1, \ldots, r_A}$.

Such a scenario is therefore determined by the positive integers $n, k_1, \ldots, k_n, r_1, \ldots, r_n$.
The simplicial complex representing this scenario is
\[\Sigma_{n,\vec{k},\vec{r}} \;\;\;\;:=\;\;\;\; \Disjcopy{k_1}{r_1} \sjoin \cdots \sjoin \Disjcopy{k_n}{r_n} \Mdot\]
For example, as already mentioned in Section \ref{ssec:structural-tripartiteexample},
our main example measurement scenario,
the tripartite simplicial complex $\Disj_2 \sjoin \Disj_2 \sjoin \Disj_2$
where we want to lump together the second and third sites,
is written as \[\Sigma_{n=2,k_1=2,k_2=2,r_1=1,r_2=2} \;\;\;\;=\;\;\;\; \Disj_2 \sjoin \Disj_2^{\sjoin 2} \Mdot\]
Other examples were also provided in the previous section.


On such a scenario, we have a symmetry that identifies the appropriate copies or microscopic sites that are to be lumped together. It identifies the measurements:
\begin{align*}
a^{(1)}_j \sim \cdots \sim a^{(r_A)}_j \;\;&\;\; (\forall j \in \enset{1, \dots, k_A}), \\
b^{(1)}_j \sim \cdots \sim b^{(r_B)}_j \;\;&\;\; (\forall j \in \enset{1, \dots, k_B}), \\
\text{etc.}
\end{align*}
Formally, this symmetry is given by an action of the group $S_{r_1} \times \cdots \times S_{r_n}$, 
where $S_l$, the symmetric group on $l$ elements, acts on each of the factors corresponding to `macroscopic' sites.
We are interested in knowing under which conditions the quotient of (the semiregularisation of)
$\Sigma_{n,\vec{k},\vec{r}}$ by this symmetry is acyclic.

\begin{proposition}\label{prop:nkr-vorobev}
  The quotient of the measurement scenario $\sr(\Sigma_{n,\vec{k},\vec{r}})$ by the symmetry above
  is acyclic 
   iff one of the following holds:
   \begin{enumerate}[label=(\roman*)]
     \item\label{item:nkr-condAll}   
       each site has at least as many microscopic sites or copies as it has measurement settings,
       i.e. $\Forall{i \in \enset{1, \ldots, n}} k_i \leq r_i$;
     \item\label{item:nkr-condAllButOne} one of the sites has a single copy
       and the condition above is satisfied by all the other sites,
       i.e. $\Exists{i_0} \left(r_{i_0} = 1 \;\land\; \Forall{i \in \enset{1,\ldots \widehat{i_0} \ldots,n}} k_i \leq r_i\right)$.
   \end{enumerate}
\end{proposition}
\begin{proof}
See \cite{RSB:DPhil-thesis-forth} for the proof of this result.
The examples of Sections \ref{ssec:structural-tripartiteexample} and \ref{ssec:structural-nonacyclicexample}
provide some intuition.
\end{proof}

The way in which this proposition splits into two cases might strike one as strange at first sight.
The first case is better suited for a reading of the result in terms of macroscopic averages,
whereas the second case resembles more closely the usual monogamy relations, where one deals with the correlations
shared by a single party with several others. 
As mentioned in Section \ref{sec:introduction},
we can read the result of Proposition \ref{prop:nkr-vorobev} as a generalisation of the results of
Ramanathan et al. \cite{RamanathanEtAl:LocalRealismOfMacroscopicCorrelations} and Paw{\l}owski \& Brukner \cite{PawlowskiBrukner:MonogamyOfBellIneqsInNonsigTheories}.

From the former's perspective, suppose that we have several (a large number of) microsystems distributed over $n$ sites, with $r_i$ microsystems at site $i$.
The group $S_{r_1} \dirprod \cdots \dirprod S_{r_n}$ captures the symmetry of the system: we can interchange any of the $r_i$
microsystems within the same site $i$.
Now assume there are $k_i$ measurement settings available at each site $i$. Microscopically, we need to consider
$k_i$ possible measurements for each microsystem.
But we consider that, macroscopically, only the average behaviour is accessible, with the corresponding
measurements being lumped together as $k_i$ averaged measurements.
The fact that the quotient is acyclic as long as there are enough microsystems in each site means the following:
no matter what the statistics for all the original microscopic measurements are (as long as they satisfy no-signalling),
the average behaviour is classical, in the sense that it admits a local hidden variable description.
This generalises the paper's result because it holds for any no-signalling theory and not just for quantum mechanics.

Note, however, that by augmenting the number of (macroscopic) measurements that one performs,
it would in principle be possible to detect non-locality on the average macroscopic correlations.
However, this soon becomes impractical if one has a large (say $\approx\!10^{23}$) number of microsystems in each site.
So, it seems that the limitation on our experimental capability of performing enough measurements
makes the average behaviour appear local.

From the point of view of Paw{\l}owski \& Brukner \cite{PawlowskiBrukner:MonogamyOfBellIneqsInNonsigTheories}, we start with an
$n$-partite scenario with $k_i$ measurement settings for each site $i$.
Then the question is: fixing the first site (or any other for that matter),
how many copies of the other sites do we need to consider
so that the monogamy relation for the violation of any $n$-partite Bell-type inequality holds?
(See equation \eqref{eq:monogamymulti} for the general form of such a monogamy relation.)
That is, with how many copies of the other sites can Alice violate the same Bell-type inequality?
The authors of the paper consider only the case $n=2$ and show that one can take $k_2$ copies of the second site in order to get the monogamy relations.
Our proposition above generalises this for any $n$,
giving the correct monogamy relation for this general case.

Moreover, our proposition is a complete characterisation: 
not only does it say that it suffices to take $k_i$ copies of each site $i$, 
it also says that taking less than that is not enough. That is, if one takes less copies of some site,
there exists a no-signalling empirical model
that violates the monogamy relations.
Similarly, the interpretation in terms of locality of macroscopic averages is also an equivalence. 
This is another way in which our result generalises both papers.

\section{Conclusions and outlook}\label{sec:conclusions}
This work explores a connection between monogamy of non-locality and the locality of average macroscopic behaviour in multipartite scenarios.
We show that both can be explained by a structural property of the simplicial complex representing the compatibility of measurements
in the scenario: after taking a quotient by an appropriate symmetry along which one takes the average or considers the monogamy relation,
the resulting complex should be acyclic, hence inherently local or non-contextual according to \Vorobev's theorem.
This means, in particular, that the proof is independent of quantum mechanics and works more generally for any no-signalling theory.
In the present document, we have motivated and illustrated the main ideas behind this analysis
via some simple example measurement scenarios.

The language of simplicial complexes, as used in the sheaf-theoretic framework \cite{AbramskyBrandenburger},
allows one to describe not only the Bell-type multipartite scenarios familiar from discussions of non-locality that we have been considering,
but also more general contextuality scenarios, such as Kochen--Specker configurations \cite{KochenSpecker}.
In upcoming work, we develop a scheme formalising our analysis in this more general setting.
The result for Bell-type scenarios stated in Proposition \ref{prop:nkr-vorobev},
whose full proof will also appear there,
can be seen as a first instance or application of that scheme.
Future work includes applying this scheme in different kinds of scenarios to yield 
monogamy relations for contextuality inequalities and to study non-contextuality of macroscopic averages.

\section*{Acknowledgements}
I thank Samson Abramsky, Adam Brandenburger, and Shane Mansfield
for valuable guidance, discussions, and comments on several versions of this work.
I also thank Miguel Navascu\'es for some very important clarifications.
Finally, I thank audiences of the seminars at Paris Diderot and ParisTech for their helpful feedback.

I gratefully acknowledge support from the
Marie Curie Initial Training Network MALOA -- From MAthematical LOgic to Applications, PITN-GA-2009-238381,
and from FCT -- Funda\c{c}\~{a}o para a Ci\^{e}ncia e Tecnologia (the Portuguese Foundation for Science and Technology),
PhD grant SFRH/BD/94945/2013.

\bibliographystyle{eptcs}
\bibliography{refs}

\begin{thebibliography}{10}
\providecommand{\bibitemdeclare}[2]{}
\providecommand{\urlprefix}{Available at }
\providecommand{\url}[1]{\texttt{#1}}
\providecommand{\href}[2]{\texttt{#2}}
\providecommand{\urlalt}[2]{\href{#1}{#2}}
\providecommand{\doi}[1]{doi:\urlalt{http://dx.doi.org/#1}{#1}}
\providecommand{\bibinfo}[2]{#2}

\bibitemdeclare{incollection}{Abramsky12:databases}
\bibitem{Abramsky12:databases}
\bibinfo{author}{Samson Abramsky} (\bibinfo{year}{2013}):
  \emph{\bibinfo{title}{Relational databases and {B}ell's theorem}}.
\newblock In \bibinfo{editor}{Val Tannen}, \bibinfo{editor}{Limsoon Wong},
  \bibinfo{editor}{Leonid Libkin}, \bibinfo{editor}{Wenfei Fan},
  \bibinfo{editor}{Wang-Chiew Tan} \& \bibinfo{editor}{Michael Fourman},
  editors: {\sl \bibinfo{booktitle}{In search of elegance in the theory and
  practice of computation: Essays dedicated to {P}eter {B}uneman}}, {\sl
  \bibinfo{series}{Lecture Notes in Computer Science}} \bibinfo{volume}{8000},
  \bibinfo{publisher}{Springer Berlin Heidelberg}, pp. \bibinfo{pages}{13--35},
  \doi{10.1007/978-3-642-41660-6\_2}.
\newblock \bibinfo{note}{Eprint available at
  {\href{http://arxiv.org/abs/1208.6416}{arXiv:1208.6416 [cs.LO]}}}.

\bibitemdeclare{article}{AbramskyBrandenburger}
\bibitem{AbramskyBrandenburger}
\bibinfo{author}{Samson Abramsky} \& \bibinfo{author}{Adam Brandenburger}
  (\bibinfo{year}{2011}): \emph{\bibinfo{title}{The sheaf-theoretic structure
  of non-locality and contextuality}}.
\newblock {\sl \bibinfo{journal}{New Journal of Physics}}
  \bibinfo{volume}{13}(\bibinfo{number}{11}), p. \bibinfo{pages}{113036},
  \doi{10.1088/1367-2630/13/11/113036}.
\newblock \bibinfo{note}{Eprint available at
  {\href{http://arxiv.org/abs/1102.0264}{arXiv:1102.0264 [quant-ph]}}}.

\bibitemdeclare{inproceedings}{AbramskyConstantin:ClassificationMultipartiteStates}
\bibitem{AbramskyConstantin:ClassificationMultipartiteStates}
\bibinfo{author}{Samson Abramsky} \& \bibinfo{author}{Carmen Constantin}
  (\bibinfo{year}{2014}): \emph{\bibinfo{title}{A classification of
  multipartite states by degree of non-locality}}.
\newblock In \bibinfo{editor}{Bob Coecke} \& \bibinfo{editor}{Matty Hoban},
  editors: {\sl \bibinfo{booktitle}{Proceedings 10th International Workshop on
  Quantum Physics and Logic, Castelldefels (Barcelona), Spain, 17th to 19th
  July 2013}}, {\sl \bibinfo{series}{Electronic Proceedings in Theoretical
  Computer Science}} \bibinfo{volume}{171}, \bibinfo{publisher}{Open Publishing
  Association}, pp. \bibinfo{pages}{10--25}, \doi{10.4204/EPTCS.171.2}.

\bibitemdeclare{article}{AbramskyHardy:LogicalBellIneqs}
\bibitem{AbramskyHardy:LogicalBellIneqs}
\bibinfo{author}{Samson Abramsky} \& \bibinfo{author}{Lucien Hardy}
  (\bibinfo{year}{2012}): \emph{\bibinfo{title}{Logical {B}ell inequalities}}.
\newblock {\sl \bibinfo{journal}{Physical Review A}}
  \bibinfo{volume}{85}(\bibinfo{number}{6}), p. \bibinfo{pages}{062114},
  \doi{10.1103/PhysRevA.85.062114}.
\newblock \bibinfo{note}{Eprint available at
  {\href{http://arxiv.org/abs/1203.1352}{arXiv:1203.1352 [quant-ph]}}}.

\bibitemdeclare{inproceedings}{AbramskyMansfieldBarbosa:Cohomology-QPL}
\bibitem{AbramskyMansfieldBarbosa:Cohomology-QPL}
\bibinfo{author}{Samson Abramsky}, \bibinfo{author}{Shane Mansfield} \&
  \bibinfo{author}{Rui Soares~Barbosa} (\bibinfo{year}{2012}):
  \emph{\bibinfo{title}{The cohomology of non-locality and contextuality}}.
\newblock In \bibinfo{editor}{Bart Jacobs}, \bibinfo{editor}{Peter Selinger} \&
  \bibinfo{editor}{Bas Spitters}, editors: {\sl \bibinfo{booktitle}{Proceedings
  8th International Workshop on Quantum Physics and Logic, Nijmegen,
  Netherlands, October 27--29, 2011}}, {\sl \bibinfo{series}{Electronic
  Proceedings in Theoretical Computer Science}}~\bibinfo{volume}{95},
  \bibinfo{publisher}{Open Publishing Association}, pp. \bibinfo{pages}{1--14},
  \doi{10.4204/EPTCS.95.1}.
\newblock \bibinfo{note}{Eprint available at
  {\href{http://arxiv.org/abs/1111.3620}{arXiv:1111.3620 [quant-ph]}}}.

\bibitemdeclare{article}{BancalEtAl2008}
\bibitem{BancalEtAl2008}
\bibinfo{author}{Jean-Daniel Bancal}, \bibinfo{author}{Cyril Branciard},
  \bibinfo{author}{Nicolas Brunner}, \bibinfo{author}{Nicolas Gisin},
  \bibinfo{author}{Sandu Popescu} \& \bibinfo{author}{Christoph Simon}
  (\bibinfo{year}{2008}): \emph{\bibinfo{title}{Testing a {B}ell inequality in
  multipair scenarios}}.
\newblock {\sl \bibinfo{journal}{Physical Review A}}
  \bibinfo{volume}{78}(\bibinfo{number}{6}), p. \bibinfo{pages}{062110},
  \doi{10.1103/PhysRevA.78.062110}.
\newblock \bibinfo{note}{Eprint available at
  {\href{http://arxiv.org/abs/0810.0942}{arXiv:0810.0942 [quant-ph]}}}.

\bibitemdeclare{article}{BarrettEtAl2005:NonlocalCorrelationsInfoResource}
\bibitem{BarrettEtAl2005:NonlocalCorrelationsInfoResource}
\bibinfo{author}{Jonathan Barrett}, \bibinfo{author}{Noah Linden},
  \bibinfo{author}{Serge Massar}, \bibinfo{author}{Stefano Pironio},
  \bibinfo{author}{Sandu Popescu} \& \bibinfo{author}{David Roberts}
  (\bibinfo{year}{2005}): \emph{\bibinfo{title}{Nonlocal correlations as an
  information-theoretic resource}}.
\newblock {\sl \bibinfo{journal}{Physical Review A}}
  \bibinfo{volume}{71}(\bibinfo{number}{2}), p. \bibinfo{pages}{022101},
  \doi{10.1103/PhysRevA.71.022101}.
\newblock \bibinfo{note}{Eprint available at
  {\href{http://arxiv.org/abs/quant-ph/0404097}{arXiv:quant-ph/0404097}}}.

\bibitemdeclare{article}{Bell-thm}
\bibitem{Bell-thm}
\bibinfo{author}{John~S. Bell} (\bibinfo{year}{1964}): \emph{\bibinfo{title}{On
  the {E}instein-{P}odolsky-{R}osen paradox}}.
\newblock {\sl \bibinfo{journal}{Physics}}
  \bibinfo{volume}{1}(\bibinfo{number}{3}), pp. \bibinfo{pages}{195--200}.

\bibitemdeclare{article}{KochenSpecker}
\bibitem{KochenSpecker}
\bibinfo{author}{Simon Kochen} \& \bibinfo{author}{Ernst~P. Specker}
  (\bibinfo{year}{1967}): \emph{\bibinfo{title}{The problem of hidden variables
  in quantum mechanics}}.
\newblock {\sl \bibinfo{journal}{Journal of Mathematics and Mechanics}}
  \bibinfo{volume}{17}(\bibinfo{number}{1}), pp. \bibinfo{pages}{59--87}.

\bibitemdeclare{phdthesis}{Mansfield:DPhil-thesis}
\bibitem{Mansfield:DPhil-thesis}
\bibinfo{author}{Shane Mansfield} (\bibinfo{year}{2013}):
  \emph{\bibinfo{title}{The mathematical structure of non-locality and
  contextuality}}.
\newblock \bibinfo{type}{{DPhil} thesis}, \bibinfo{school}{University of
  Oxford}.

\bibitemdeclare{inproceedings}{MansfieldBarbosa:QPL2013}
\bibitem{MansfieldBarbosa:QPL2013}
\bibinfo{author}{Shane Mansfield} \& \bibinfo{author}{Rui Soares~Barbosa}
  (\bibinfo{year}{2013}): \emph{\bibinfo{title}{Extendability in the
  sheaf-theoretic approach: Construction of {B}ell models from
  {K}ochen-{S}pecker models}}.
\newblock In: {\sl \bibinfo{booktitle}{(Pre)Proceedings of 10th Wokshop on
  Quantum Physics and Logic (QPL X), ICFo Barcelona}}.
\newblock \bibinfo{note}{Eprint available at
  {\href{http://arxiv.org/abs/1402.4827}{arXiv:1402.4827 [quant-ph]}}}.

\bibitemdeclare{article}{NavascuesPironioAcin2008}
\bibitem{NavascuesPironioAcin2008}
\bibinfo{author}{Miguel Navascu{\'e}s}, \bibinfo{author}{Stefano Pironio} \&
  \bibinfo{author}{Antonio Ac{\'i}n} (\bibinfo{year}{2008}):
  \emph{\bibinfo{title}{A convergent hierarchy of semidefinite programs
  characterizing the set of quantum correlations}}.
\newblock {\sl \bibinfo{journal}{New Journal of Physics}}
  \bibinfo{volume}{10}(\bibinfo{number}{7}), p. \bibinfo{pages}{073013},
  \doi{10.1088/1367-2630/10/7/073013}.
\newblock \bibinfo{note}{Eprint available at
  {\href{http://arxiv.org/abs/0803.4290}{arXiv:0803.4290 [quant-ph]}}}.

\bibitemdeclare{article}{NavascuesWunderlich2009}
\bibitem{NavascuesWunderlich2009}
\bibinfo{author}{Miguel Navascu{\'e}s} \& \bibinfo{author}{Harald Wunderlich}
  (\bibinfo{year}{2009}): \emph{\bibinfo{title}{A glance beyond the quantum
  model}}.
\newblock {\sl \bibinfo{journal}{Proceedings of the Royal Society A}}
  \bibinfo{volume}{466}(\bibinfo{number}{2115}), pp. \bibinfo{pages}{881--890},
  \doi{10.1098/rspa.2009.0453}.
\newblock \bibinfo{note}{Eprint available at
  {\href{http://arxiv.org/abs/0907.0372}{arXiv:0907.0372 [quant-ph]}}}.

\bibitemdeclare{article}{PawlowskiBrukner:MonogamyOfBellIneqsInNonsigTheories}
\bibitem{PawlowskiBrukner:MonogamyOfBellIneqsInNonsigTheories}
\bibinfo{author}{Marcin Paw{\l}owski} \& \bibinfo{author}{\ifmmode
  \check{C}\else~\v{C}\fi{}aslav Brukner} (\bibinfo{year}{2009}):
  \emph{\bibinfo{title}{Monogamy of {B}ell's inequality violations in
  nonsignaling theories}}.
\newblock {\sl \bibinfo{journal}{Physical Review Letters}}
  \bibinfo{volume}{102}(\bibinfo{number}{3}), p. \bibinfo{pages}{030403},
  \doi{10.1103/PhysRevLett.102.030403}.
\newblock \bibinfo{note}{Eprint available at
  {\href{http://arxiv.org/abs/0810.1175}{arXiv:0810.1175 [quant-ph]}}}.

\bibitemdeclare{article}{PR-boxes}
\bibitem{PR-boxes}
\bibinfo{author}{Sandu Popescu} \& \bibinfo{author}{Daniel Rohrlich}
  (\bibinfo{year}{1994}): \emph{\bibinfo{title}{Quantum nonlocality as an
  axiom}}.
\newblock {\sl \bibinfo{journal}{Foundations of Physics}}
  \bibinfo{volume}{24}(\bibinfo{number}{3}), pp. \bibinfo{pages}{379--385},
  \doi{10.1007/BF02058098}.

\bibitemdeclare{article}{RamanathanEtAl:LocalRealismOfMacroscopicCorrelations}
\bibitem{RamanathanEtAl:LocalRealismOfMacroscopicCorrelations}
\bibinfo{author}{Ravishankar Ramanathan}, \bibinfo{author}{Tomasz Paterek},
  \bibinfo{author}{Alastair Kay}, \bibinfo{author}{Pawel
  Kurzy\ifmmode~\acute{n}\else \'{n}\fi{}ski} \& \bibinfo{author}{Dagomir
  Kaszlikowski} (\bibinfo{year}{2011}): \emph{\bibinfo{title}{Local realism of
  macroscopic correlations}}.
\newblock {\sl \bibinfo{journal}{Physical Review Letters}}
  \bibinfo{volume}{107}(\bibinfo{number}{6}), p. \bibinfo{pages}{060405},
  \doi{10.1103/PhysRevLett.107.060405}.
\newblock \bibinfo{note}{Eprint available at
  {\href{http://arxiv.org/abs/1010.2016}{arXiv:1010.2016 [quant-ph]}}}.

\bibitemdeclare{misc}{RSB:DPhil-thesis-forth}
\bibitem{RSB:DPhil-thesis-forth}
\bibinfo{author}{Rui Soares~Barbosa}: \bibinfo{note}{Forthcoming {DPhil}
  thesis, University of Oxford}.

\bibitemdeclare{article}{vorobev}
\bibitem{vorobev}
\bibinfo{author}{Nikolai~N. Vorob'ev} (\bibinfo{year}{1962}):
  \emph{\bibinfo{title}{Consistent families of measures and their extensions}}.
\newblock {\sl \bibinfo{journal}{Theory of Probability and its Applications
  (Teoriya Veroyatnostei i ee Primeneniya)}}
  \bibinfo{volume}{7}(\bibinfo{number}{2}), pp. \bibinfo{pages}{147--163
  (English: N. Greenleaf, trans.), 153--159 (Russian)}, \doi{10.1137/1107014}.

\end{thebibliography}

\end{document}